\documentclass[a4paper,fleqn]{cas-sc}

\usepackage[numbers, sort&compress]{natbib}
\usepackage{subfigure}
\usepackage{capt-of} 
\usepackage{amsthm}
\newtheorem{theorem}{Theorem}
\newtheorem{corollary}{Corollary}
\theoremstyle{remark}
\newtheorem{remark}{Remark}
 \definecolor{lightblue}{rgb}{0.68, 0.85, 0.9}
 \definecolor{lightgreen}{rgb}{0.73, 0.91, 0.62}
 \definecolor{lightorange}{rgb}{0.98, 0.84, 0.65}
 \usepackage{algorithmic}
 \usepackage{algorithm}
  \newcommand{\mytablestyle}{
    \rowcolors{2}{white}{white}
    \arrayrulecolor{lightblue!50!black}
  }
\def\tsc#1{\csdef{#1}{\textsc{\lowercase{#1}}\xspace}}
\tsc{WGM}
\tsc{QE}
\tsc{EP}
\tsc{PMS}
\tsc{BEC}
\tsc{DE}

\begin{document}
\let\WriteBookmarks\relax
\def\floatpagepagefraction{1}
\def\textpagefraction{.001}
\let\printorcid\relax 
\shorttitle{Applied Mathematical Modelling}
\shortauthors{Jiexi Tang et~al.}

\title [mode = title]{SIS Epidemic Modelling on Homogeneous Networked System: General Recovering Process and Mean-Field Perspective}                      
\tnotemark[1]

\tnotetext[1]{This work is supported by the National Natural Science Foundation of China (NSFC) (GrantNo.~62206230) and the Natural Science Foundation of Chongqing (GrantNo.~CSTB2023NSCQ-20MSX0064).}

\author[1]{Jiexi Tang} 

\author[1]{Yichao Yao} 

\author[1]{Meiling Xie} 

\author[1]{Minyu Feng} 
\cormark[1]

\credit{Conceptualization of this study, Methodology, Software}

\address[1]{College of Artificial Intelligence, Southwest University, Chongqing, 400715, China}

\ead[URL]{myfeng@swu.edu.cn}

\cortext[cor1]{Corresponding author}


  \begin{abstract}
    Although we have made progress in understanding disease spread in complex systems with non-Poissonian activity patterns, current models still fail to capture the full range of recovery time distributions. In this paper, we propose an extension of the classic susceptible-infected-susceptible (SIS) model, called the general recovering process SIS (grp-SIS) model. This model incorporates arbitrary recovery time distributions for infected nodes within the system. We derive the mean-field equations assuming a homogeneous network, provide solutions for specific recovery time distributions, and investigate the probability density function (PDF) for infection times in the system’s steady state. Our findings show that recovery time distributions significantly affect disease dynamics, and we suggest several future research directions, including extending the model to arbitrary infection processes and using the quasistationary method to address deviations in numerical results.
    \end{abstract}


\maketitle

\section{Introduction}
\label{sec:1}

Recent studies on epidemic dynamics in complex networks highlight the role of higher-order interactions, non-pharmaceutical interventions (e.g., quarantine and mobility control), and structural network features \cite{chen2024siqrs,li2024urban,cui2024impact}. While traditional models are useful, they often fail to capture the full complexity of real-world scenarios, especially in social and multiplex networks where classical centrality measures inadequately represent node influence \cite{hajarathaiah2024node,cisneros2021multigroup}. Recent works have examined both stochastic and deterministic epidemic models, optimal resource allocation, and the effects of inter-event time distributions on disease dynamics, yielding novel insights into epidemic spreading and control \cite{de2024certain,jafarizadeh2025optimal,papageorgiou2023stochastic,zhang2021threshold}. These insights extend to digital epidemics, where models help analyze malware propagation in network security \cite{sangala2024analyzing}. Furthermore, researchers have optimized epidemic control strategies by identifying critical nodes in networks to minimize spread \cite{zhao2020minimum}. Moreover, recent studies have introduced coupled epidemic-information propagation models that incorporate individual adaptive behavior and information diffusion, significantly shaping epidemic dynamics \cite{an2024coupled}. Network-based approaches further contribute to understanding epidemic spreading, such as modelling disease transmission across areas using mobility-driven complex networks \cite{li2022network} and analyzing information dynamics in evolving networks under birth-death processes \cite{feng2024information}.

Research into the distribution of inter-event times has unveiled significant insights into the dynamics of complex systems, with power law distributions observed in the recurrence of extreme events, underscoring their irregular yet predictable nature \cite{benson2007recurrence}. This understanding is further enhanced by considering burstiness and network topology in spreading dynamics, as non-stationary patterns significantly impact the propagation of epidemics and information \cite{horvath2014spreading,jo2014analytically}. The critical role of small inter-event times in epidemic spreading on networks has been widely recognized, highlighting the necessity of incorporating these temporal dynamics into modelling \cite{masuda2020small,min2011spreading}. Models that account for the non-Poissonian nature of human activities offer a more precise representation of spreading dynamics \cite{vazquez2006modeling}. Distributed Radio Frequency Identification (RFID) sensor networks enable detailed analyses of person-to-person interactions, further clarifying the dynamics of human contact patterns \cite{cattuto2010dynamics}. Human activity patterns significantly influence information diffusion, underscoring the importance of accounting for non-stationarity and burstiness \cite{iribarren2009impact,aliee2020estimating}. Similarly, \cite{zhu2019nonlinear} analyzed the incorporation of time delay and nonlinear dynamics into epidemic modelling, providing insights into the stability and control of diseases within network-based systems. This understanding is crucial for developing models that accurately reflect the complexities of real-world disease dynamics and information spread.

Previous studies have investigated the spreading of the non-Markovian epidemic, where bursty or heavy-tailed waiting times significantly impact disease propagation. Min et al. \cite{min2011spreading} examined SI dynamics with power law waiting times and observed a slow decay in prevalence, albeit under the assumption of fixed distributions that limit generalizability. Masuda and Holme \cite{masuda2020small} demonstrated that short inter-event times exert a greater influence on epidemic thresholds than long-tailed distributions, although their framework is confined to mixed exponentials. Jo et al. \cite{hathcock2022asymptotic} investigated the effects of burstiness during both early and late stages of spreading, while their analysis predominantly focused on initial dynamics. Cator et al. \cite{cator2013susceptible} introduced the generalized SIS (GSIS) model that allows for arbitrary infection and recovery times and derived mean-field equations similar to NIMFA; however, their study did not address transient dynamics or systematically compare recovery distributions. Starnini et al. \cite{starnini2017equivalence} mapped non-Markovian processes onto equivalent Markovian models using an effective infection rate, thereby simplifying the analysis; however, this approach assumes that memory effects can be fully captured by a single parameter, an assumption that may not hold under strong temporal correlations. Feng et al. \cite{feng2019equivalence} further demonstrated that Markovian approximations break down in the presence of temporal correlations in edge activation, leading to transient and steady-state prediction deviations. Han et al. \cite{han2023non} studied non-Markovian epidemic spreading on temporal networks and found that, under certain conditions, these dynamics can be approximated by mapping them onto Markovian processes on static networks; Zeng et al. \cite{zeng2025power} proposed a complex network model with power-law activating patterns, uncovering stationary distributions and evolutionary dynamics through renewal theory and Markov analysis. However, their focus was primarily on topological effects rather than on the nuances of infection-recovery dynamics.

While these works provide valuable theoretical insights, they either rely on Markovian approximations, emphasize steady-state properties, or neglect transient behavior. In this paper, we introduce the General Recovering Process SIS (grp-SIS) model, a significant advancement of the conventional SIS model. This enhanced model accounts for arbitrary recovery time distributions, providing a more practical framework for simulating disease dynamics within intricate systems. Differing from the classical SIS model, the grp-SIS model does not rely on the assumption of a memoryless recovery process. This departure is necessitated by the often unrealistic nature of such an assumption, considering factors such as disease severity and individual immunity levels present within a system. Tailored for homogeneous network systems, our model delves into the effects of diverse recovery time distributions on the spread of disease throughout the system. Unlike preceding methods that either rely on Markovian approximations or predominantly focus on steady-state behavior, our model directly resolves the complete mean-field equations governing non-Markovian SIS dynamics. This approach allows for an accurate depiction of both transient and steady-state behaviors. Through the integration of varied recovery distributions without resorting to simplifying assumptions, this approach not only enhances prediction accuracy but also ensures computational feasibility. Beyond its theoretical contributions, our model holds substantial practical value in fields such as public health, healthcare resource allocation, vaccine development, and network security, offering a robust analytical tool for understanding and mitigating the spread of infections and threats in various complex systems.

The structure of this paper is as follows: Section~\ref{sec:1} provides a comprehensive overview of the intricate dynamics of disease transmission within complex systems, along with the pivotal influence exerted by inter-event time distributions across these networks. Section~\ref{sec:2} delves into a detailed review of the relevant literature, exploring the foundational models that underpin our understanding. Following this, we introduce the innovative grp-SIS model, specifically designed for complex systems, and derive the mean-field equations governing this intricate system. Subsequently, we conduct an in-depth examination of the infection and recovery dynamics under diverse recovery time distributions within the system. Section~\ref{sec:4} shifts its focus to numerical simulations, comparing the theoretical predictions with empirical observations in the context of system behavior. Lastly, Section~\ref{sec:5} summarizes the key findings and outlines potential avenues for future research in the study of complex systems.


\section{Modelling of the SIS model with a general recovering process}
\label{sec:2}

In the study of disease dynamics within network systems, the classical SIS model has been widely used due to its simplicity and mathematical tractability. However, this model relies on certain assumptions, such as the memoryless recovery process governed by a Poisson process, which may not hold in more complex real-world system scenarios. As real-world data often exhibit non-exponential recovery time distributions, a variety of examples highlight this phenomenon across different system contexts. For instance, in healthcare systems, recovery times for illnesses such as chronic infections or certain viral diseases are often highly variable and do not follow a simple exponential trend; patients may experience prolonged recovery with sporadic relapses. Similarly, in online systems like social media platforms, user activity often follows a bursty pattern, where engagement surges sporadically due to external triggers, resulting in non-uniform return times that align more closely with heavy-tailed distributions. Another example can be observed in natural ecosystems, where populations of certain species recover or replenish at irregular intervals due to fluctuating environmental factors or predator-prey dynamics. These cases illustrate that recovery processes in real-world systems are frequently influenced by complex, context-specific factors, leading to recovery time distributions that deviate significantly from the exponential model. To address this variability, we propose a more flexible and realistic SIS model capable of accommodating arbitrary recovery time distributions within network systems. By allowing for diverse recovery patterns, our model better reflects the stochastic and often unpredictable nature of real-world spreading processes, providing a more accurate framework for studying epidemic dynamics within these systems. This new model, referred to as the grp-SIS model, relaxes the classical assumptions and allows for the investigation of how different recovery time distributions affect disease spread within network systems.

\subsection{Classical SIS model}
\label{sec:2.1}

The classical SIS model has been foundational in the study of disease dynamics on networks. In this model, individuals transition between susceptible and infected states, with the dynamics governed by two key processes: infection and recovery. The mean-field estimation method provides the propagation dynamics equation for the SIS model as follows \cite{kermack1932contributions}:
%
\begin{align}
\label{eq:classic-sis}
\frac{\text{d}\rho _{I}(t)}{\text{d}t}=-\mu \rho _{I}(t)+\beta \langle k
\rangle [1-\rho _{I}(t)]\rho _{I}(t).
\end{align}
Here, $\rho _{I}(t)$ represents the fraction of infected individuals in the population at time $t$. The parameter $\mu $ is the recovery rate, with units of $\text{s}^{-1}$, indicating the probability per unit time that an infected individual recovers and returns to the susceptible state. The parameter $\beta $ is the infection rate, also with units of $\text{s}^{-1}$, representing the probability per unit time that an infected node transmits the disease to a susceptible neighbor. The term $\langle k \rangle $ denotes the average degree of the network, which characterizes the average number of connections per node. This model assumes that the infection process is memoryless, implying that the time until recovery follows an exponential distribution, leading to the Markovian nature of the SIS dynamics. Since Eq. {(\ref{eq:classic-sis})} is a separable differential equation, we can integrate both sides under the given initial conditions $(\rho _{I}(0), \rho _{S}(0)) = (\rho _{I}^{0}, \rho _{S}^{0})$ to obtain the analytical solution:
%
\begin{align}
\label{eq:classic-sis-analytical}
\rho _{I}(t)=
\frac{\beta \langle k\rangle -\mu}{\beta \langle k\rangle +\frac{\beta \langle k\rangle \rho _{S}^{0}-\mu}{\rho _{I}^{0}}\text{e}^{(\beta \langle k\rangle -\mu )t}}
\forall t\geq 0,\tau >\tau _{c},
\end{align}
where $\tau _{c}=\frac{1}{\langle k\rangle}$ is the epidemic threshold, which characterizes the critical point above which the disease can persist in the population.

For a robust assessment of the classical SIS model's applicability, we need to estimate its parameters using real data. We can use recovery waiting time data for each infected individual and infection waiting time data for each susceptible individual, denoted as $\{\mathbf{W}_{i}^{I}, \mathbf{W}_{j}^{S}\}^{n, m}_{i=1,j=1}$. Given our assumptions, these two times should follow exponential distributions with parameters $\mu $ and $\beta \langle k\rangle $. Maximum likelihood estimation can be applied to estimate these parameters as follows:
%
\begin{align}
\begin{cases}
\displaystyle \max _{\mu}\left \{\prod _{i}\mu \exp (-\mu \mathbf{W}_{i}^{I})
\right \}
\\\noalign{\vspace{3pt}}
\displaystyle \max _{\beta}\left \{\prod _{j}\beta \langle k\rangle
\exp (-\beta \langle k\rangle \mathbf{W}_{j}^{S})\right \}
\end{cases}.
\label{eq3}
\end{align}
These yield the estimates:
%
\begin{align}
\begin{cases}
\displaystyle \hat{\mu}=\frac{1}{n}\sum _{i}\mathbf{W}_{i}^{I}=:
\overline{\mathbf{W}^{I}}
\\\noalign{\vspace{3pt}}
\displaystyle \hat{\beta}\langle k\rangle =\frac{1}{m}\sum _{j}
\mathbf{W}_{j}^{S}=:\overline{\mathbf{W}^{S}}
\end{cases},
\label{eq4}
\end{align}
where $\overline{\mathbf{W}^{I}}$ and $\overline{\mathbf{W}^{S}}$ represent the sample means of the recovery and infection waiting times, respectively. Through this process, we obtain the maximum likelihood estimates of the exponential distributions for the two waiting times: $\hat{f}_{W^{I}}(t)$ for recovery waiting time and $\hat{f}_{W^{S}}(t)$ for infection waiting time.

Next, we divide the data into $N$ equal-sized intervals $\{E^{I}_{i},E^{S}_{i}\}_{i=1}^{N}$ based on the data range, resulting in $N$ datasets $\{\mathscr{W}_{i}^{I},\mathscr{W}_{i}^{S}\}_{i=1}^{N}$. Here, $\mathscr{W}_{i}^{I}$ and $\mathscr{W}_{i}^{S}$ denote the subsets of waiting time data that fall within each interval $E^{I}_{i}$ and $E^{S}_{i}$, respectively. To estimate the probability density function (PDF) directly, we employ the Kernel Density Estimation (KDE) method:
%
\begin{align}
\begin{cases}
f_{{W}^{I}}^{*}(t) = \dfrac{|\mathscr{W}^{I}_{i}|}{|E^{I}_{i}|n},&
\forall t\in E^{I}_{i}
\\\noalign{\vspace{3pt}}
f_{{W}^{S}}^{*}(t) = \dfrac{|\mathscr{W}^{S}_{i}|}{|E^{S}_{i}|m},&
\forall t\in E^{S}_{i}
\end{cases}.
\label{eq5}
\end{align}
Finally, we compute the Kullback-Leibler (KL) divergence between the estimated distributions and the expected exponential distribution with the corresponding parameters. KL divergence is a widely used measure for quantifying the difference between two probability distributions, making it an appropriate tool for assessing how well the observed waiting time distributions align with the assumed exponential model. The KL divergence between the estimated PDF $f^{*}(t)$ and the theoretical exponential PDF $\hat{f}(t)$ is given by:
%
\begin{align}
D_{\text{KL}}(f^{*} || \hat{f}) = \int _{0}^{\infty }f^{*}(t) \ln
\frac{f^{*}(t)}{\hat{f}(t)} \text{d}t.
\label{eq6}
\end{align}
By evaluating this divergence, we can determine the extent to which the empirical data deviates from the exponential assumption, providing insights into the validity of the classical SIS model.

\subsection{Assumptions of the grp-SIS model}
\label{sec:2.2}

When the KL divergence is large, the classical SIS model fails to accurately describe systemic disease transmission behavior within complex systems. A large KL divergence indicates that our data is not exponentially distributed, which conflicts with the assumptions of the classical SIS model. The assumption that the remaining recovering time $W-T$ of a patient in the SIS model is independent of the previous infection time $T$ is unreasonable within the context of a system. In real life, we can roughly determine the recovery time interval by experience through the recovery momentum of a patient.

Because of the analysis above, we propose an SIS model with a general recovery waiting time distribution (hereinafter referred to as the grp-SIS model), which is based on the following assumptions:
\begin{enumerate}[3)]
\item[1)] The infection process of each infected node for the susceptible neighbor node is a Poisson process with intensity $\beta $. Here, $\beta $ has units of $\text{s}^{-1}$, representing the probability per unit time that an infected node transmits the disease.
\item[2)] The recovering process of the infected node itself is a counting process derived from several distribution $N(t)$, but with limitations: $\forall t\geq 0$,
%
\begin{align}
P\{N(t+\text dt)\geq 1,N(t)=0\}=\mathcal{O}(\text dt).
\label{eq7}
\end{align}
\item[3)] The recovery waiting time for a node is a random variable $W$, assuming its PDF is $w(t)$. Here are the requirements:
%
\begin{align}
w(t)\in \mathcal{R}(\mathbb R);\forall t<0, w(t)=0.
\label{eq8}
\end{align}
\item[4)] Each recovery waiting time is independent of all previous recovery/infection waiting times. In other words, the current infection status of an individual is not significantly related to the number of times they have been previously infected.
\end{enumerate}

Then, we can express this {grp-SIS model ({Fig.~\ref{fig:grp-SIS_sketch_map}})} as follows:
%
\begin{align}
\begin{cases}
I+S\mathop{\rightarrow}\limits ^{\beta}2I
\\\noalign{\vspace{3pt}}
I\mathop{\rightarrow}\limits ^{w(t)}S
\end{cases}.
\label{eq9}
\end{align}
\begin{figure}
    \centering
    \subfigure[Infection process]{
        \includegraphics[width=0.7\columnwidth]{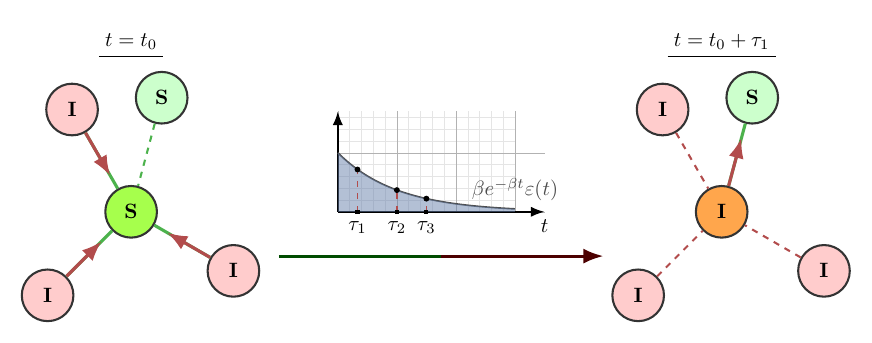}
    }
    \subfigure[Recovery process]{
        \includegraphics[width=0.7\columnwidth]{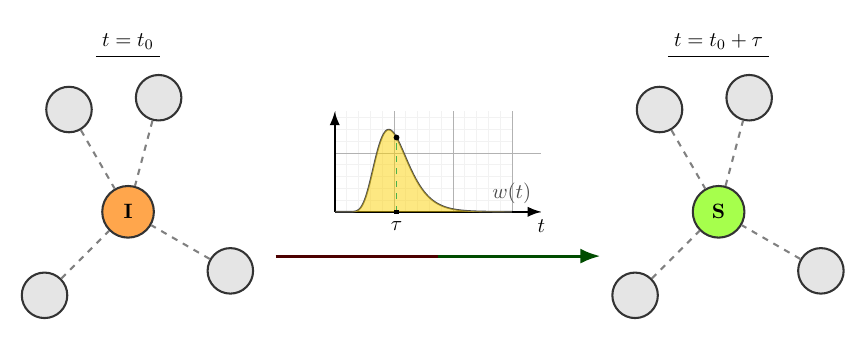}
    }
    \captionof{figure}{
        \textbf{Schematic diagram of the grp-SIS model.} 
        Red nodes represent infected nodes, green nodes indicate susceptible nodes, 
        and gray nodes denote nodes whose states are arbitrary. The focal nodes for each process 
        are highlighted with a darker color. Dashed edges represent connections where the nodes 
        at both ends do not influence the ongoing process of the focal node, whereas solid 
        directed edges from red to green nodes indicate the potential transmission of infection along that link.  
        The diagram (a) illustrates the infection process of a susceptible node in 
        the grp-SIS model. Initially, at $t = t_0$, a susceptible node (S) is 
        surrounded by both infected (I) and susceptible (S) nodes, with directed 
        edges indicating potential transmission from the infected nodes. The middle 
        section depicts the probability density function (PDF) of the infection waiting 
        time, with labeled points $\tau_1$, $\tau_2$, and $\tau_3$ representing 
        individual times required for transmission attempts from different infected 
        neighbors. At $t = t_0 + \tau_1$, the susceptible node becomes infected, 
        altering the network's state and the potential interactions.
        The diagram (b) shows the recovery process of an infected node in the grp-SIS 
        model. At $t = t_0$, an infected node (I) is connected to surrounding nodes 
        with no significant interaction for recovery. The middle section includes a 
        PDF representing the distribution of recovering waiting times, with the 
        labeled point $\tau$ denoting the specific recovery time for the infected 
        node. At $t = t_0 + \tau$, the node transitions to a susceptible state, 
        depicted by a change in color, indicating recovery and potential reinfection 
        by neighboring nodes.
    } 
    
    \label{fig:grp-SIS_sketch_map}
\end{figure}

Using homogeneous networks in the study of complex network spreading systems offers key advantages. They simplify the analysis and modelling of these systems, allowing for easier derivation and clearer theoretical results, such as threshold conditions and steady-state solutions. Homogeneous networks provide a standardized foundation for validating complex models within spreading systems and facilitate better control and prediction of spreading processes. They are also ideal for initial research to explore fundamental mechanisms before applying findings to more heterogeneous networks \cite{daley1999epidemic}. Therefore, in this study, we use homogeneous networks for modelling spreading systems.

\subsection{Survival probability in the recovery process}
\label{sec:2.3}

To derive the dynamic equations of the grp-SIS model, we first describe the counting process derived from the general waiting time. In this context, $ F_{0}(t) $ represents the probability that no event (i.e., no recovery event) occurs up to time $ t $; in other words, it is the probability that an individual remains infected until at least time $ t $.

Assume that the waiting time $ W $ between adjacent events in the counting process $ N(t) $ has a probability density function $ w(t) $, satisfying $ w(t)=0 $ for all $ t<0 $. Then, we have
%
\begin{align}
F_{0}(t) = P\{N(t)=0\} = P\{W>t\} = \int _{t}^{\infty }w(\tau ) \, {
\mathrm{d}}\tau .
\label{eq10}
\end{align}

Since, for disease transmission, we only care about whether an individual is infected at a given time (i.e., whether at least one recovery event has not occurred), the quantity $ F_{0}(t) $ plays a crucial role in our subsequent analysis. Empirical studies \cite{gandica2017stationarity,barabasi2005origin,vazquez2007impact,jo2012optimized,oliveira2005darwin,gonccalves2008human,candia2008uncovering,radicchi2009human,chaintreau2007impact} have shown that in various real-world scenarios, the inter-event time distribution exhibits heavy-tailed characteristics, often following a power law distribution. Such distributions arise in human activity patterns, social interactions, and biological processes, making them highly relevant for modelling disease recovery dynamics. Additionally, the lognormal distribution is frequently employed to describe inter-event time distributions, as it effectively captures variability in recovery times due to heterogeneous factors such as individual immunity, treatment effects, and disease severity \cite{crow1987lognormal}. \cite{hathcock2022asymptotic} further highlights other possibilities, suggesting that the choice of recovery time distribution significantly impacts epidemic dynamics. Therefore, in this study, we also use power law and lognormal distributions as examples. Below, we present $ F_{0}(t) $ for these classic recovery time distributions.
\begin{enumerate}[3)]
\item[1)] Exponential distribution:
%
\begin{align}
w(t)=\mu {\mathrm{{e}}}^{-\mu t}\varepsilon (t)
\label{eq11}
\end{align}
and
%
\begin{align}
F_{0}(t)={\mathrm{{e}}}^{-\mu t}\varepsilon (t)
\label{eq12}
\end{align}
where $\mu >0$.
\item[2)] Power law distribution:
%
\begin{align}
w(t)=\displaystyle \frac{\lambda -1}{t_{0}}\left (\frac{t}{t_{0}}
\right )^{-\lambda}\varepsilon (t-t_{0})
\label{eq13}
\end{align}
and
%
\begin{align}
F_{0}(t)=\displaystyle \left (\frac{t}{t_{0}}\right )^{1-\lambda}
\varepsilon (t-t_{0})
\label{eq14}
\end{align}
where $\lambda >2,t_{0}>0$.
\item[3)] Lognormal distribution:
%
\begin{align}
w(t)=
\begin{cases}
\frac{1}{t\sqrt{2\pi}\sigma}\exp \left ({-
\frac{(\ln t - \mu )^{2}}{2\sigma ^{2}}}\right ),&t>0
\\\noalign{\vspace{3pt}}
0,&t\leq 0
\end{cases}
\label{eq15}
\end{align}
and
%
\begin{align}
F_{0}(t)=
\begin{cases}
\frac{1}{2}\text{erfc}\left (\frac{\ln t -\mu}{\sqrt 2\sigma}\right ),&t>0
\\\noalign{\vspace{3pt}}
0,&t\leq 0
\end{cases}
\label{eq16}
\end{align}
where $\text{erfc}(x)$ is the co-error {function,}
\end{enumerate}
where the function $\varepsilon (t)$ is the unit step function:
%
\begin{align}
\varepsilon (t)=
\begin{cases}
1,&t\geq 0
\\\noalign{\vspace{3pt}}
0,&t<0
\end{cases}.
\label{eq17}
\end{align}
%

\subsection{Dynamic equations of the grp-SIS model on homogeneous networks}
\label{sec:2.4}

In this section, we derive the homogeneous mean-field equations of the grp-SIS model, assuming the disease spreads on a homogeneous network with an average degree $\langle k\rangle $ within a complex system.

First, we consider the state variables for establishing this model within the system. Since the general distribution no longer possesses the memoryless property like the exponential distribution, we need to consider the state information of each infected node at every time point before the current time within the system. According to assumption $4$ of the grp-SIS model, we only need to trace back to the most recent infection of the infected node within the system, which means we only need to consider the infection duration $T(t)$ of the infected node. Therefore, we define:
%
\begin{align}
\rho _{I}(t;\tau ):=P\{I,T(t)\leq \tau \}.
\label{eq18}
\end{align}
The probability that a node has been infected for no more than $\tau $ at time $t$ within the system. For the state variable of susceptible individuals, since we still assume the infection process is a Poisson process, we use the symbol $\rho _{S}(t)$ to denote the susceptible density at time $t$ within the system. As for the infection density at time $t$, we have the following relationship:
%
\begin{align}
\rho _{I}(t)=\lim _{\tau \rightarrow \infty}\rho _{I}(t;\tau ),
\forall t\geq 0.
\label{eq19}
\end{align}
Concerning $\rho _{I}(t; \tau )$, the properties are as follows:
%
\begin{align}
\label{eq:pItt}
\begin{cases}
\displaystyle \rho _{I}(t;\tau )=\rho _{I}(t)\int _{0}^{\tau }f_{T(t)}(
\tau ')\text{d}\tau ', & \forall t\geq 0
\\\noalign{\vspace{3pt}}
\displaystyle \rho _{I}(t;\tau )=0,&\forall \tau \leq 0
\end{cases},
\end{align}
where the function $f_{T(t)}(\tau )$ is the one-dimensional probability density function of the stochastic process $T(t)$ within the system.

For the homogeneous mean-field equations of the grp-SIS model within the system, we have the following theorem:

\begin{theorem}
\label{thm1}
Let $\rho _{S}(t)$ and $\rho _{I}(t)$ denote the susceptible and infected densities at time $t$, respectively, with the constraint $\rho _{S}(t) + \rho _{I}(t) = 1$. For the infected nodes, let $w(t)$ be the probability density function of the recovery waiting time $W^{I}$, satisfying the grp-SIS model conditions. The probability of a node not recovering after time $\tau $ is given by $F_{0}(\tau )$.

Then, the mean-field equations for the grp-SIS spread on a homogeneous network with an average degree $\langle k\rangle $ are:
%
\begin{align}
\label{eq:grp-sis}
\frac{\partial \rho _{I}(t;\tau )}{\partial t}+
\frac{\partial \rho _{I}(t;\tau )}{\partial \tau}=\beta \langle k
\rangle \rho _{S}(t)\rho _{I}(t)\varepsilon (\tau )-\int _{0}^{\tau}\frac{w(\tau ')}{F_{0}(\tau ')}
\frac{\partial \rho _{I}(t;\tau ')}{\partial \tau '}\text{d}\tau '.
\end{align}
\end{theorem}
\begin{proof}
Assume that at time $t$, the infection duration of an infected node is $T(t)$, the total recovery waiting time of the node is $W^{I}$, and the infection waiting time of a susceptible node is $W^{S}$. We then classify and prove the different state changes of nodes.

1) Assume that at time $t$, a node is susceptible. We need to find the probability that the node becomes infected after a time interval
$\text{d}t$.

The node has an average of $\langle k\rangle \rho _{I}(t)$ infected neighboring nodes. First, we consider the case where the node gets infected in the interval $(t, t+\text{d}t]$ and remains infected until time $t+\text{d}t$. Due to the memoryless property of the exponential distribution, this case can be regarded as if the node has just recovered at time $t$, without considering its previous state. Let the region $D_{1}:=\{(\tau _{1},\tau _{2})\in \mathbb R_{+}^{2}|\tau _{1}\leq \text{d}t,\tau _{1}+\tau _{2}>\text{d}t\}$, then the probability is
%
\begin{align}
&P\{W^{S}\leq \text{d}t,W^{S}+W^{I}>\text{d}t\}
=\iint _{D_{1}}w(\tau _{1})\beta \langle k\rangle \rho _{I}(t)
\text{e}^{-\beta \langle k\rangle \rho _{I}(t)\tau _{2}}\text{d}\tau _{1}
\text{d}\tau _{2}.
\label{eq22}
\end{align}
Let the region
$D_{2}^{\tau}:=\{(\tau _{1},\tau _{2})\in \mathbb R_{+}^{2}|\tau <
\tau _{1}\leq \tau +\text{d}t,\tau _{1}+\tau _{2}\leq \tau +\text{d}t\}$ Then
we know that there is
%
\begin{align}
&\iint _{D_{1}+D_{2}^{0}}w(\tau _{1})\beta \langle k\rangle \rho _{I}(t)
\text{e}^{-\beta \langle k\rangle \rho _{I}(t)\tau _{2}}\text{d}\tau _{1}
\text{d}\tau _{2}
\notag
\\
&=\int _{\mathbb R_{+}}w(\tau _{1})\text{d}\tau _{1}\int _{0}^{\text{d}t}
\beta \langle k\rangle \rho _{I}(t)\text{e}^{-\beta \langle k\rangle
\rho _{I}(t)\tau _{2}}\text{d}\tau _{2}
\notag
\\
&=1-\text{e}^{-\beta \langle k\rangle \rho _{I}(t) \text{d}t}
\notag
\\
&=\beta \langle k\rangle \rho _{I}(t) \text{d}t + \mathcal{O}(\text{d}t^{2}).
\label{eq23}
\end{align}
Also, for the region $D_{2}^{\tau}$, let
%
\begin{align}
E := [\tau , \tau +\text{d}t] \times [0, \text{d}t],
\label{eq24}
\end{align}
which denotes the Cartesian product of the closed intervals $[\tau , \tau +\text{d}t]$ and $[0, \text{d}t]${, using} assumption, we have
%
\begin{align}
&\iint _{D_{2}^{\tau}}w(\tau _{1})\beta \langle k\rangle \rho _{I}(t)
\text{e}^{-\beta \langle k\rangle \rho _{I}(t)\tau _{2}}\text{d}\tau _{1}
\text{d}\tau _{2}
\notag
\\
&\leq \iint _{E}w(\tau _{1})\beta \langle k\rangle \rho _{I}(t)
\text{e}^{-\beta \langle k\rangle \rho _{I}(t)\tau _{2}}\text{d}\tau _{1}
\text{d}\tau _{2}
\notag
\\
&=\mathcal{O}(\text{d}t)\int _{\tau}^{\tau +\text{d}t} w(\tau _{1})
\text{d}\tau _{1}
\notag
\\
&=\mathcal{O}(\text{d}t)P\{N(\tau +\text{d}t)\geq 1,N(\tau )=0\}
\notag
\\
&=\mathcal{O}(\text{d}t^{2}).
\label{eq25}
\end{align}
Here, the last step follows from assumption 2) of the model. So, obviously:
%
\begin{align}
&\iint _{D_{1}}w(\tau _{1})\beta \langle k\rangle \rho _{I}(t)\text{e}^{-
\beta \langle k\rangle \rho _{I}(t)\tau _{2}}\text{d}\tau _{1}\text{d}
\tau _{2}=\beta \langle k\rangle \rho _{I}(t) \text{d}t + \mathcal{O}(\text{d}t^{2}).
\label{eq26}
\end{align}

For the remaining cases, it can be asserted that their probability is $\mathcal{O}(\text{d}t^{2})$. This is because, in any case, the node will be infected at least once and recover within the time interval $(t, t +\text{d}t]$, hence $W^{S} + W^{I} \leq \text{d}t$. Therefore, in the remaining cases, we must have $(W^{I}, W^{S}) \in D_{2}^{0}$, thus
%
\begin{align}
&P\{\text{remaining cases}\}\leq \iint _{D_{2}^{0}} w(\tau _{1}) \beta \langle k\rangle \rho _{I}(t)
\text{e}^{-\beta \langle k\rangle \rho _{I}(t) \tau _{2}} \text{d}\tau _{1}
\text{d}\tau _{2}
= \mathcal{O}(\text{d}t^{2}).
\label{eq27}
\end{align}
Additionally, since an infected node can't have an infection time less than $0$, we need to multiply by the unit step function. Therefore, the final probability is $\beta \langle k\rangle \rho _{I}(t)\varepsilon (\tau ) \text{d}t + \mathcal{O}(\text{d}t^{2})$.

2) At time $t$, a node has been infected for exactly $\tau $ time. We need to find the probability that the node remains infected after a time interval $\text{d}t$.

Assume that the node has not healed until time $t+\text{d}t$.

The probability of this situation occurring is
%
\begin{align}
&P\{W^{I}-T(t)>\text{d}t|W^{I}>\tau ,T(t)=\tau \}
\notag
\\
&=
\frac{P\{W^{I}-T(t)>\text{d}t,W^{I}>\tau |T(t)=\tau \}}{P\{W^{I}>\tau \}}
\notag
\\
&=\frac{P\{W^{I}>\tau + \text{d}t\}}{P\{W^{I}>\tau \}}=
\dfrac{\int _{\tau +\text{d}t}^{\infty }w(t)\text{d}t}{\int _{\tau}^{\infty }w(t)\text{d}t}
\notag
\\
&=\dfrac{F_{0}(\tau +\text{d}t)}{F_{0}(\tau )}
=1-\dfrac{w(\tau )}{F_{0}(\tau )}\text{d}t+o(\text{d}t),a.e.
\label{eq28}
\end{align}

Other situations involve at least one susceptible interval appearing within the interval $[t, t+\text{dt}]$. It can be asserted that the probability of these situations is $\mathcal{O}(\text{d}t^{2})$. This is because, in the remaining cases, we must have $W^{I} + W^{S} \leq \tau + \text{d}t$ and $W^{I} > \tau $, so we must have $(W^{S}, W^{I}) \in D_{2}^{\tau}$, thus
%
\begin{align}
&P\{\text{remaining cases}\}
\leq \iint _{D_{2}^{\tau}}w(\tau _{1})\beta \langle k\rangle \rho _{I}(t)
\text{e}^{-\beta \langle k\rangle \rho _{I}(t)\tau _{2}}\text{d}\tau _{1}
\text{d}\tau _{2}
=\mathcal{O}(\text{d}t^{2}).
\label{eq29}
\end{align}

Therefore, the final probability is $1-\dfrac{w(\tau )}{F_{0}(\tau )}\text{d}t+o(\text{d}t),a.e${.}

In addition, by the property {(\ref{eq:pItt})} we can know that
%
\begin{align}
\label{eq:T-pI}
f_{T(t)}(\tau )=\dfrac{1}{\rho _{I}(t)}
\dfrac{\partial \rho _{I}(t;\tau )}{\partial \tau}.
\end{align}

Therefore, from the total probability formula, the probability that the node at the time of $t+\text{d}t$ is infected not more than $\tau +\text{d}t$ is:
%
\begin{align}
&\rho _{I}(t+\text{d}t;\tau +\text{d}t)
=\left (\beta \langle k\rangle \rho _{I}(t)\varepsilon (\tau )
\text{d}t\right )\rho _{S}(t)
 +\rho _{I}(t)\int _{0}^{\tau}\left (1-
\dfrac{w(\tau ')}{F_{0}(\tau ')}\text{d}t\right )
\dfrac{1}{\rho _{I}(t)}
\dfrac{\partial \rho _{I}(t;\tau )}{\partial \tau}\text{d}\tau '
+o(\text{d}t).
\label{eq31}
\end{align}
Subtract $\rho _{I}(t; \tau )$ and set $\text{d}t\rightarrow 0$, we get
%
\begin{align}
&\lim _{\text{d}t\rightarrow 0}
\frac{\rho _{I}(t+\text{d}t;\tau +\text{d}t)-\rho _{I}(t;\tau )}{\text{d}t}=
\frac{\partial \rho _{I}}{\partial t}+
\frac{\partial \rho _{I}}{\partial \tau}
=\beta \langle k\rangle \rho _{S}(t)\rho _{I}(t)\varepsilon (\tau )-
\int _{0}^{\tau}\frac{w(\tau ')}{F_{0}(\tau ')}
\frac{\partial \rho _{I}(t;\tau ')}{\partial \tau '}\text{d}\tau '.
\label{eq32}
\end{align}
That completes the proof.
\end{proof}
\begin{figure}[htbp]
    \centering
    \subfigure[Susceptible node transition]{
        \includegraphics[width=0.6\linewidth]{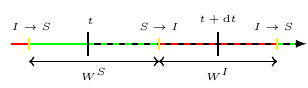}
    }

    \subfigure[Infected node transition]{
        \includegraphics[width=0.6\linewidth]{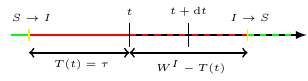}
    }

    \subfigure[Alternative infected node transition]{
        \includegraphics[width=0.6\linewidth]{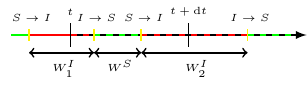}
    }
\captionof{figure}{
    \textbf{Schematic diagram of state transitions in nodes under different scenarios.} The timeline 
    illustrates the changes in node states over time, with green segments denoting the 
    susceptible state and red segments indicating the infected state. Dashed lines demonstrate 
    potential future states. In the timeline, black vertical lines mark the boundaries of the time 
    interval \((t, t+\text dt]\) considered in the derivation, 
    while yellow vertical lines indicate the actual moments of state transitions. 
    Therefore, the timeline segments on both sides of a black vertical line have the same color, 
    while those on both sides of a yellow vertical line have different colors, reflecting the state change at that moment. 
    The top timeline (a) shows the most likely scenario when the current 
    node state is susceptible. The middle timeline (b) depicts the most probable scenario when the 
    current node state is infected. The bottom timeline (c) illustrates other possible scenarios 
    when the current node state is infected, excluding the scenario shown in the middle timeline (b). 
    } 

  \label{fig:state-change-diagram}
  \end{figure}
%
\begin{remark}
\label{rem1}
The integral term in Eq. {(\ref{eq:grp-sis})},
%
\begin{align}
Q(t,\tau ) := \int _{0}^{\tau} \frac{w(\tau ')}{F_{0}(\tau ')}
\frac{\partial \rho _{I}(t;\tau ')}{\partial \tau '} \, \text d\tau ',
\label{eq33}
\end{align}
represents the cumulative effect of recovery processes on the infected population over time. Specifically, the term $\frac{w(\tau ')}{F_{0}(\tau ')}$ denotes the instantaneous hazard rate, which quantifies the probability density of recovery occurring at $\tau '$ given that no recovery has happened before $\tau '$. The integrand $\frac{\partial \rho _{I}(t;\tau ')}{\partial \tau '}$ captures the rate at which infected individuals accumulate at different infection ages $\tau '$, accounting for the continuous progression of infection within the population.

Physically, $Q(t,\tau )$ describes the total influence of the recovery process on individuals who have been infected for a duration $\tau $, effectively measuring how the distribution of infection ages impacts the overall dynamics of disease persistence. It characterizes the rate at which individuals with an infection age of at most $\tau $ contribute to the depletion of the infected class due to recovery. This term ensures that the model properly accounts for non-exponential recovery patterns, distinguishing it from classical Markovian SIS models.
\end{remark}

\begin{remark}
\label{rem2}
In Eq. {(\ref{eq:grp-sis})}, if we set $\tau \rightarrow \infty $, we can get the equation for infection density $\rho _{I}(t)$:
%
\begin{align}
\label{eq:special-eq}
\frac{{\mathrm{{d}}}\rho _{I}(t)}{{\mathrm{{d}}}t}=\beta \langle k\rangle \rho _{S}(t)
\rho _{I}(t)-\int _{\mathbb R_{+}}\frac{w(\tau )}{F_{0}(\tau )}
\frac{\partial \rho _{I}(t;\tau )}{\partial \tau}{\mathrm{{d}}}\tau.
\end{align}
But this equation still contains $\rho _{I}(t; \tau )$ term, which makes it difficult to use Eq. {(\ref{eq:special-eq})} to solve the infection density in general. But we can note that if $w(\tau )$ is chosen so that the equation $\frac{w(\tau )}{F_{0}(\tau )}= constant $ is satisfied, this means that the recovery waiting time follows an exponential distribution, so that the integral on the right side of the equation just simplifies, which makes Eq. {(\ref{eq:special-eq})} become Eq. {(\ref{eq:classic-sis})}{. This} is exactly the homogeneous mean-field equation of the classical SIS model.
\end{remark}
%
\begin{remark}
\label{rem3}
In general, Eq. {(\ref{eq:grp-sis})} is not very convenient for solving problems. We can take the partial derivative of $\tau $ on both sides of the equation at the same time to get the following equation:
%
\begin{align}
\label{eq:grp-sis-plus}
{\frac{\partial ^{2}\rho _{I}}{\partial \tau ^{2}}+
\frac{\partial ^{2}\rho _{I}}{\partial t\partial \tau}+
\frac{w(\tau )}{F_{0}(\tau )}\frac{\partial \rho _{I}}{\partial \tau}=
\beta \langle k\rangle \rho _{S}(t)\rho _{I}(t)\delta (\tau )},
\end{align}
where $\delta (\tau )$ is the unit impulse function. This is a second-order non-homogeneous linear partial differential equation with variable coefficients, which can be combined with Eq. {(\ref{eq:grp-sis-plus})} and initial infection density distribution $(\rho _{I}(0),\rho _{S}(0))$, initial infection time distribution $\rho _{I}(0; \tau )$ These three conditions form a definite solution problem to be solved.
\end{remark}

As an application of {Theorem~\ref{thm1}}, we solve the solution under the classical SIS system. At this time there are $w (\tau ) = \mu \text{e} ^ {- \mu \tau} $, and $\rho _{I} (t) $ is given by Eq. {(\ref{eq:classic-sis-analytical})}. For convenience, we use the symbol $g(t,\tau )$ to denote $\beta \langle k\rangle \rho _{S}(t)\rho _{I}(t)\varepsilon (\tau )$. This gives us a first-order linear partial differential equation:
%
\begin{align}
\frac{\partial \rho _{I}(t;\tau )}{\partial t}+
\frac{\partial \rho _{I}(t;\tau )}{\partial \tau}+\mu \rho _{I}(t;
\tau )=g(t,\tau ).
\label{eq36}
\end{align}
Given the initial infection time distribution $\rho _{I}(0; \tau )$, we can get the special solution of the equation under this problem:
%
\begin{align}
\rho _{I}(t;\tau )=&\beta \langle k\rangle \int _{0}^{\min (t,\tau )}
\rho _{S}(t-\xi )\rho _{I}(t-\xi )\text{e}^{-\mu \xi}\text{d}\xi
+\rho _{I}(0;\tau -t)\text{e}^{-\mu t}.
\label{eq37}
\end{align}

\subsection{Statistical characteristics}
\label{sec:3}

This subsection will present some statistical characteristics of the grp-SIS model. To obtain the infection density at the steady state of propagation, we first need to describe the results related to the time already infected, $T(\infty )$, of infected nodes at steady state.

\begin{theorem}
\label{tho:PDF-T}
If $\mathbb{E}[W]$ exists, the probability density function of $T(\infty )$ in the non-absorbing state is given by the following expression:
%
\begin{align}
f_{T(\infty )}(\tau )=\dfrac{F_{0}(\tau )}{\mathbb{E}[W]}\varepsilon (
\tau ).
\label{eq38}
\end{align}
\end{theorem}
\begin{proof}
First, from Eq. {(\ref{eq:grp-sis})} and Eq. {(\ref{eq:T-pI})}, we have:
%
\begin{align}
\frac{\partial \rho _{I}(t;\tau )}{\partial t}+f_{T(t)}(\tau )\rho _{I}(t)=
\beta \langle k\rangle \rho _{S}(t)\rho _{I}(t)\varepsilon (\tau )
-\rho _{I}(t)\int _{0}^{\tau}\frac{w(\tau ')}{F_{0}(\tau ')} f_{T(t)}(
\tau ')\text{d}\tau '.
\label{eq39}
\end{align}
Letting $t\to \infty $, we have $\frac{\partial \rho _{I}(t;\tau )}{\partial t} \to 0$ and $\rho _{S}(t) \to \rho _{S}^{\infty}$. Substituting these into the above equation, we obtain:
%
\begin{align}
\label{eq:wentai}
\rho ^{\infty}_{I} \left (f_{T(\infty )}(\tau )-\beta \langle k
\rangle \rho _{S}^{\infty}\varepsilon (\tau )+\int _{0}^{\tau}
\frac{w(\tau ')}{F_{0}(\tau ')} f_{T(\infty )}(\tau ')\text{d}\tau '
\right )=0.
\end{align}
Since $\rho _{I}^{\infty }\neq 0$ in the non-absorbing state, we finally obtain the integral equation:
%
\begin{align}
\label{eq:fT}
{f_{T(\infty )}(\tau )=\beta \langle k\rangle \rho ^{\infty}_{S}
\varepsilon (\tau )-\int _{0}^{\tau}\frac{w(t)}{F_{0}(t)}f_{T(\infty )}(t)
\text{d}t}.
\end{align}
This integral equation naturally possesses an initial condition:
%
\begin{align}
f_{T(\infty )}(0) = \beta \langle k\rangle \rho ^{\infty}_{S}.
\label{eq42}
\end{align}
However, this condition contains an unknown variable $\rho _{S}^{\infty}$, so an additional condition is required to obtain a unique solution. Since the desired function is a PDF, it naturally satisfies the normalization condition.

Taking the derivative on both sides of the above equation, we obtain the differential equation:
%
\begin{align}
\frac{\text{d} f_{T(\infty )}(\tau )}{\text{d}\tau} +
\frac{w(\tau )}{F_{0}(\tau )} f_{T(\infty )}(\tau ) = 0, \quad
\forall \tau \geq 0.
\label{eq43}
\end{align}

Solving this differential equation using the two aforementioned conditions and applying assumption 3) of the grp-SIS model, we obtain the final solution:
%
\begin{align}
f_{T(\infty )}(\tau )&=
\frac{F_{0}(\tau )}{\displaystyle \int _{\mathbb R_{+}}F_{0}(\tau )\text{d}\tau}
\varepsilon (\tau )=
\frac{F_{0}(\tau )}{\left .\tau F_{0}(\tau )\right |^{\infty}_{0}-\displaystyle \int _{\mathbb R_{+}}\tau \text{d}F_{0}(\tau )}
\varepsilon (\tau )
\notag
\\
&=
\frac{F_{0}(\tau )}{\displaystyle \int _{\mathbb R_{+}}\tau w(\tau )\text{d}\tau}
\varepsilon (\tau )
=\dfrac{F_{0}(\tau )}{\mathbb{E}[W]}\varepsilon (\tau ).
\label{eq44}
\end{align}
That completes the proof.
\end{proof}
%
\begin{remark}
\label{rem4}
The result of {Theorem~\ref{tho:PDF-T}} indicates that the probability of the random variable $T(\infty )$ taking a value around $\tau $ is proportional to $F_{0}(\tau )$. This conclusion is quite intuitive: $T(\infty )$ represents the duration an infected individual has been infected in the steady state, while $F_{0}(\tau )$ denotes the probability that an individual's infection time has exceeded $\tau $.

Since we focus only on individuals whose infection duration is exactly $\tau $ at the present moment, these individuals must be those who are about to reach an infection time exceeding $\tau $. These ``lucky ones'' can be illustrated by {Fig.~\ref{fig:state-change-diagram}}(b). Therefore, the process of identifying such individuals is equivalent to selecting, from all individuals whose infection time has exceeded $\tau $, those who were infected precisely before time $\tau $.

Given that the system is in a steady state, the proportion of such individuals remains dynamically stable, which implies a constant factor. Consequently, this reasoning naturally leads to the final probability density function, which takes the form of a constant multiple of $F_{0}(\tau )$. In particular, the normalization condition uniquely determines this constant as $1/\mathbb{E}[W]$, completing the derivation of the final result.
\end{remark}
From {Theorem~\ref{tho:PDF-T}}, we can immediately derive the following conclusion:
%
\begin{corollary}
\label{cor1}
If $\mathbb{E}[W^{2}]$ exists, {the} expectation of the infected time of the infected node when the propagation is steady is:
%
\begin{align}
\mathbb E[T(\infty )]=\frac{\mathbb E[W^{2}]}{2\mathbb E[W]}.
\label{eq45}
\end{align}
\end{corollary}
%
\begin{remark}
\label{rem5}
There is a subtle fallacy here: since $W$ indicates the total time an individual experiences from the onset of infection until recovery, while $T$ only denotes the time an infected individual in the population has currently been infected and has not yet recovered. Intuitively, the total waiting time $W$ would, on average, not be less than the infection time $T$ of the infected individual. However, if the second moment of $W$ is very large or even does not exist, it turns out that $\mathbb{E}[T(\infty )]$ can be much larger than $\mathbb{E}[W]$.

One way to explain this is: if the variance of $W$ is large, it increases the probability of the occurrence of individuals with long-term infections. Since our infection process is a Poisson process, this results in a greater number of infected individuals, leading to more long-term infected individuals in the population. This causes the accumulation of long-term individuals within the population, significantly raising the average infection time of the population, even if each individual's average recovery time is relatively short.
\end{remark}
%
\begin{corollary}
\label{cor2}
The infection density of the transmission steady state is:
%
\begin{align}
\rho _{I}^{\infty}=
\begin{cases}
\displaystyle 0, & \tau < \tau _{c}
\\\noalign{\vspace{3pt}}
\displaystyle \frac{\tau - \tau _{c}}{\tau}, & \tau > \tau _{c}
\end{cases}.
\label{eq46}
\end{align}
The effective transmission rate is given by $\tau = \beta \mathbb{E}[W]$, and the transmission threshold is $\tau _{c} = \frac{1}{\langle k \rangle}$.
\end{corollary}
\begin{proof}
Consider the case where $\rho _{I}^{\infty }\neq 0$. In this case, apply Eq. {(\ref{eq:fT})} and take the limit $\tau \to \infty $. We can then solve:
%
\begin{align}
\rho _{S}^{\infty }= \frac{1}{\beta \langle k \rangle} \int _{0}^{
\infty }\frac{w(t)}{F_{0}(t)} f_{T(\infty )}(t) \, \text{d}t =
\frac{1}{\beta \langle k \rangle \mathbb{E}[W]}.
\label{eq47}
\end{align}
Thus, we have $\rho _{I}^{\infty }= 1 - \rho _{S}^{\infty }= \frac{\beta \langle k \rangle \mathbb{E}[W] - 1}{\beta \langle k \rangle \mathbb{E}[W]} = \frac{\tau - \tau _{c}}{\tau}$. This should lie within the interval $(0, 1)$, thus implying that $\tau > \tau _{c}$.

It is evident that for $\tau < \tau _{c}$, $\frac{1}{\beta \langle k \rangle \mathbb{E}[W]} > 1$, which contradicts the definition of the susceptible density. Therefore, in this case, Eq. {(\ref{eq:fT})} has no solution, and by solving Eq. {(\ref{eq:wentai})}, we obtain $\rho _{I}^{\infty }= 0$.
\end{proof}

The transmission threshold is a crucial quantity that measures the minimum effective transmission rate required for the disease to persist in the population. From the above result, we observe that for the grp-SIS model, the transmission threshold remains $\tau _{c} = \frac{1}{\langle k \rangle}$, which is consistent with the classical SIS model. This indicates that the transmission threshold of an SIS-type epidemic system is a value solely determined by the structural properties of the network, independent of the specific characteristics of the disease transmission process.

Thus, we solve for the cases of the three distributions mentioned in Section~\ref{sec:2.3}. The {Table~\ref{table:Digital_Features}} presents the complete results.
\begin{table}[!h] 
  \rowcolors{2}{white}{white}
  \arrayrulecolor{lightblue!50!black}
  \centering
  \captionof{table}{Three Statistical Characteristics} 
  \label{table:Digital_Features}
  \resizebox{1.0\columnwidth}{!}{
  \begin{tabular}{c|c|c|c} 
    \toprule[2pt]
    \rowcolor{cyan!30!white}
    \toprule
    \textbf{Distribution} & \textbf{$f_{T(\infty)}(\tau)$} & \textbf{$\mathbb{E}[T(\infty)]$} & \textbf{$\rho_S^\infty$ for $\tau > \tau_c$}\\
    \midrule
    \textbf{Exponential} & $\displaystyle \mu e^{-\mu\tau}\varepsilon(\tau)$ & $\dfrac1\mu$ & $\dfrac{\mu}{\beta\langle k\rangle}$\\
    \midrule
    \textbf{Power Law} & $\displaystyle\left\{\begin{array}{ll}
      \dfrac{\lambda-2}{\lambda-1}\dfrac{1}{t_0}\left(\dfrac{\tau}{t_0}\right)^{1-\lambda}, & \tau > t_0 \\[8pt]
      \dfrac{\lambda-2}{\lambda-1}\dfrac{1}{t_0}, & \tau \in [0,t_0] \\[8pt]
      0, & \tau < 0
    \end{array}\right.$ & $\dfrac{\lambda-2}{\lambda-3}\dfrac{t_0}{2}$ & $\dfrac{\lambda-2}{\lambda-1}\dfrac{1}{t_0}\dfrac{1}{\beta\langle k\rangle}$\\
    \midrule
    \textbf{Log-Normal} & $\displaystyle\left\{\begin{array}{ll}
      \dfrac{\text{erfc}\left(\frac{\ln \tau - \mu}{\sqrt{2}\sigma}\right)}{2e^{\mu+\frac{\sigma^2}{2}}}, & \tau > 0 \\[8pt]
      e^{-\mu-\frac{\sigma^2}{2}}, & \tau = 0 \\[8pt]
      0, & \tau < 0
    \end{array}\right.$ & $\dfrac{1}{2}e^{\mu+\frac{3\sigma^2}{2}}$ & $\dfrac{1}{\beta\langle k\rangle}e^{-\mu-\frac{\sigma^2}{2}}$\\
    \bottomrule
    \bottomrule[2pt]
  \end{tabular}}
\end{table}


\section{Numerical simulation}
\label{sec:4}

Regarding the simulation of continuous-time SIS models, various methods are available. In \cite{li2012susceptible}, a node-centric event-driven approach was used to simulate the $\varepsilon $-SIS model. Similarly, this approach can be applied to the grp-SIS model, resulting in the following node-centric event-driven grp-SIS simulation method in {Algorithm~\ref{alg:ghp-SIS}}.
%
\begin{algorithm}[!h]
    \caption{Node-Centric Event-Driven Simulation of grp-SIS Model}
    \label{alg:ghp-SIS}
    \renewcommand{\algorithmicrequire}{\textbf{Input:}}
    \renewcommand{\algorithmicensure}{\textbf{Output:}}
    
    \begin{algorithmic}[1]
        \REQUIRE $T$: Simulation time, $\beta$: Infection rate, $RWT\_DIST$: Recovering waiting time distribution, $G$: Network graph, $states$: Initial states of nodes
        \ENSURE Number of susceptible and infected nodes over time

        \STATE Initialize event list $E$ with current infected nodes' recovery events based on $RWT\_DIST$
        \WHILE{Simulation time $t < T$}
            \STATE Extract the earliest event $e$ from the event list $E$
            \STATE Update current time to the time of event $e$
            
            \IF {Event $e$ is a recovery event}
                \STATE Change the state of the node from infected to susceptible
                \STATE Remove node's future infection events from $E$
            
            \ELSIF {Event $e$ is an infection event}
                \STATE Check if the target node is susceptible
                \IF {Target node is susceptible}
                    \STATE Change the state of the node to infected
                    \STATE Schedule a new recovery event for the node based on $RWT\_DIST$ and add it to $E$
                    \FOR{Each susceptible neighbor of the infected node}
                        \STATE Schedule a new infection event for the neighbor based on infection rate $\beta$ and add it to $E$
                    \ENDFOR
                \ENDIF
            \ENDIF
        \ENDWHILE
        
        \RETURN Number of susceptible and infected nodes over time
    \end{algorithmic}
  \end{algorithm}

Using the method just introduced, we performed simulations for the three distributions mentioned in Section~\ref{sec:2.3} on three uniform regular graphs with 2500 nodes. We set the simulation time to 50, with an initial infection density of 0.3. We assumed the initially infected nodes to be newly infected. The results are as follows.

\subsection{Changes in infection density}
\label{sec:4.1}

To explore the time evolution of infection and susceptibility densities under different recovery waiting time distributions, we utilized the algorithm previously detailed. For each distribution, we selected appropriate parameter sets and conducted multiple independent simulations on a uniform regular graph. We averaged the results from 50 independent runs to generate the plots shown in {Fig.~\ref{fig:rhoinfty}}, which illustrate how the densities evolve for the three distributions: exponential, lognormal, and power law.

The plots in {Fig.~\ref{fig:rhoinfty}} show that, for all three distributions, the infection and susceptibility densities gradually converge to their theoretical steady-state values, verifying the accuracy of our theoretical predictions. However, the convergence strength decreases from {Fig.~\ref{fig:rhoinfty}}(a)-(c). In {Fig.~\ref{fig:rhoinfty}}(c), a prolonged oscillation appears during the approach to the steady state, indicating underdamped behavior, whereas {Fig.~\ref{fig:rhoinfty}}(a) and {Fig.~\ref{fig:rhoinfty}}(b) show overdamped behaviors with smoother convergence. This difference arises from the long-tailed nature of the recovery waiting time distribution in the power law case. Long-lived infected nodes tend to elevate infection levels beyond expectations, requiring a lower-than-expected susceptibility density to balance infection growth. This imbalance between infection and susceptibility densities introduces an inertia effect, leading to oscillations around the steady-state values. In contrast, higher-order small values in the distribution suppress the formation of long-lived infected nodes, resulting in more concentrated node lifetimes and faster convergence in the exponential and lognormal cases.

In addition to the steady-state analysis, we conducted a transient state analysis of the system to further validate the mean-field formulation given in Eq. {(\ref{eq:grp-sis})}. To this end, we implemented an algorithm based on a dynamic programming approach, which efficiently computes the numerical solution of the grp-SIS dynamics under the prescribed initial conditions. Specifically, the algorithm solves the partial differential-integral equation in a discretized time-$\tau $ space and yields the transient infection density $\rho _{I}^{n}(t)$ as well as the corresponding susceptible density. These numerical solutions are then superimposed on the simulation results in {Fig.~\ref{fig:rhoinfty}} as semi-transparent bands---blue for the infection density and purple for the susceptible density. As shown in the figure, the numerical solution of the dynamical equation and the direct simulation results exhibit substantial overlap, indicating strong agreement in the transient regime. Notably, the simulation-derived susceptible density is slightly lower than the numerical solution. We attribute this discrepancy to the influence of the disease-free equilibrium, which exerts a downward pull on the overall infection density in the simulations, whereas the mean-field dynamic equation primarily captures the propagation behavior in the non-absorbing state.

Additionally, we conducted a statistical analysis of the 50 independent simulation runs for each recovery time distribution. This analysis allowed us to compute the 95\% confidence interval at each time step, as well as the standard deviation, resulting in a confidence band and a standard deviation function, as depicted in {Fig.~\ref{fig:std_and_ci}}.

Observing {Fig.~\ref{fig:std_and_ci}}, we find that the width of the 95\% confidence interval remains below 0.07, while the standard deviation is nearly always less than 0.01. This indicates that, apart from the initial phase of the simulation, the fluctuations in the data are minimal throughout the process. The small variability further validates the stability and accuracy of our simulation methodology, demonstrating the robustness of our numerical approach in capturing the epidemic dynamics under different recovery time distributions.

\begin{figure*}
    \centering
    \subfigure[Exponential]{
        \includegraphics[width=0.3\linewidth]{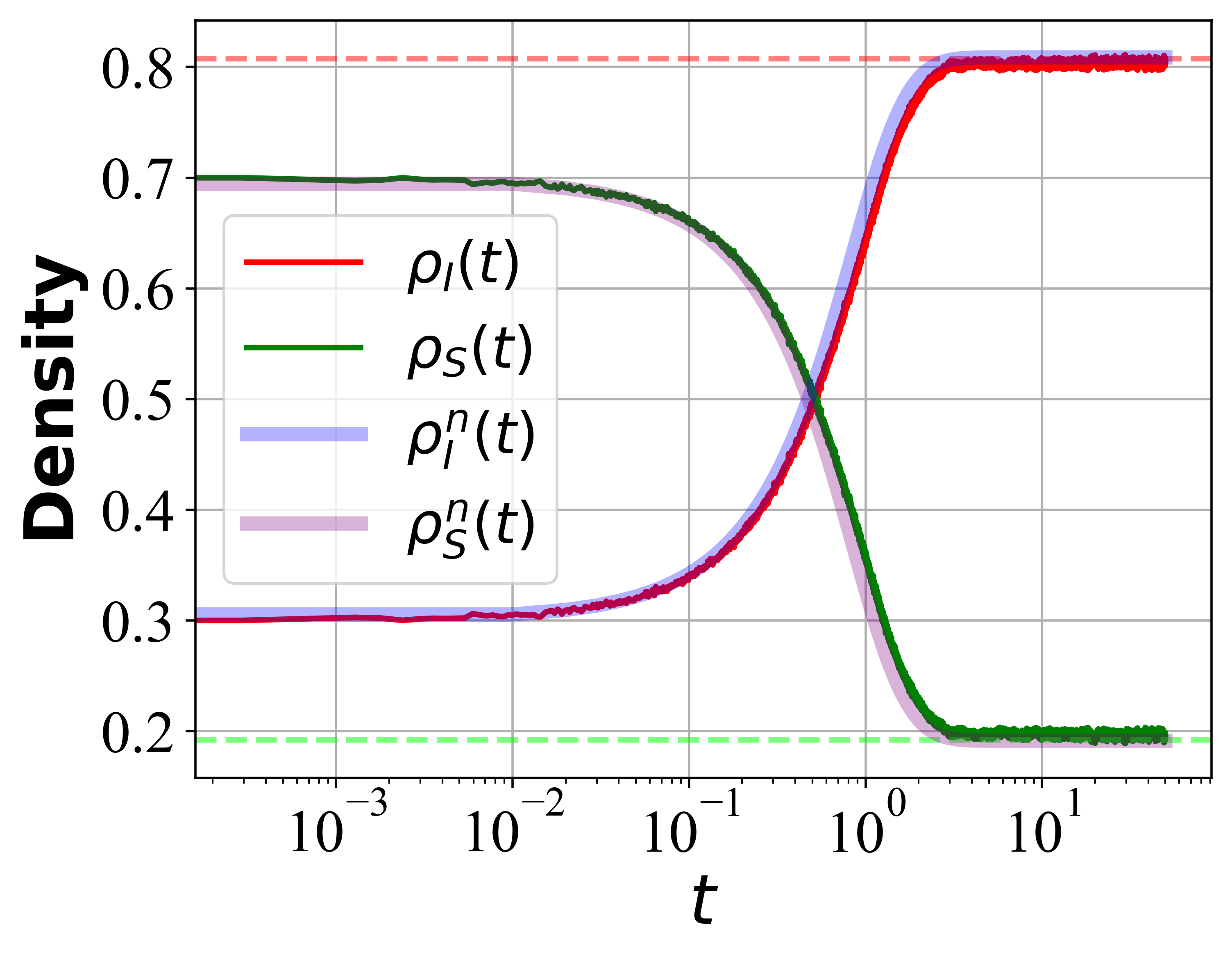}
    }
    \subfigure[Lognormal]{
        \includegraphics[width=0.3\linewidth]{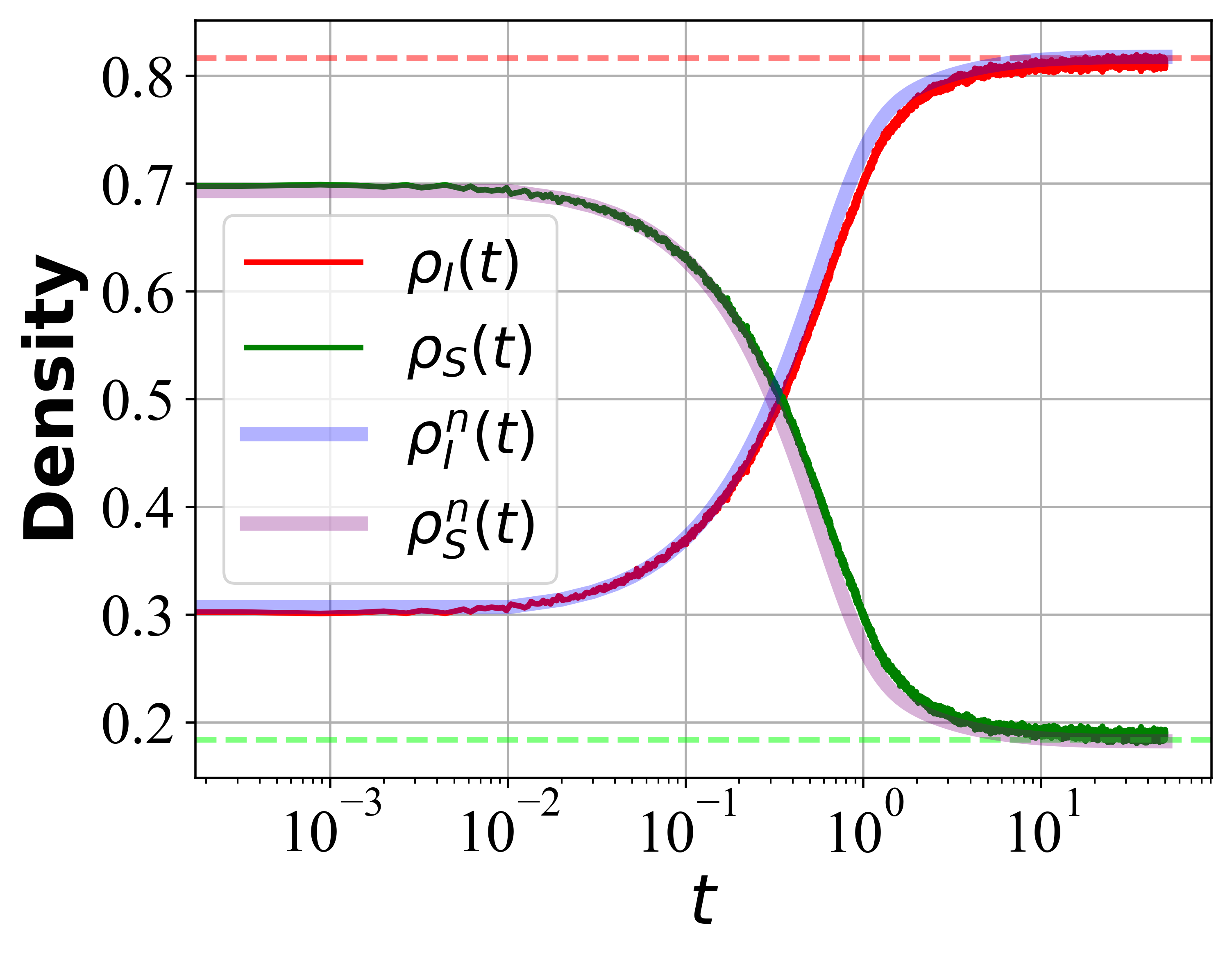}
    }
    \subfigure[Power law]{
        \includegraphics[width=0.3\linewidth]{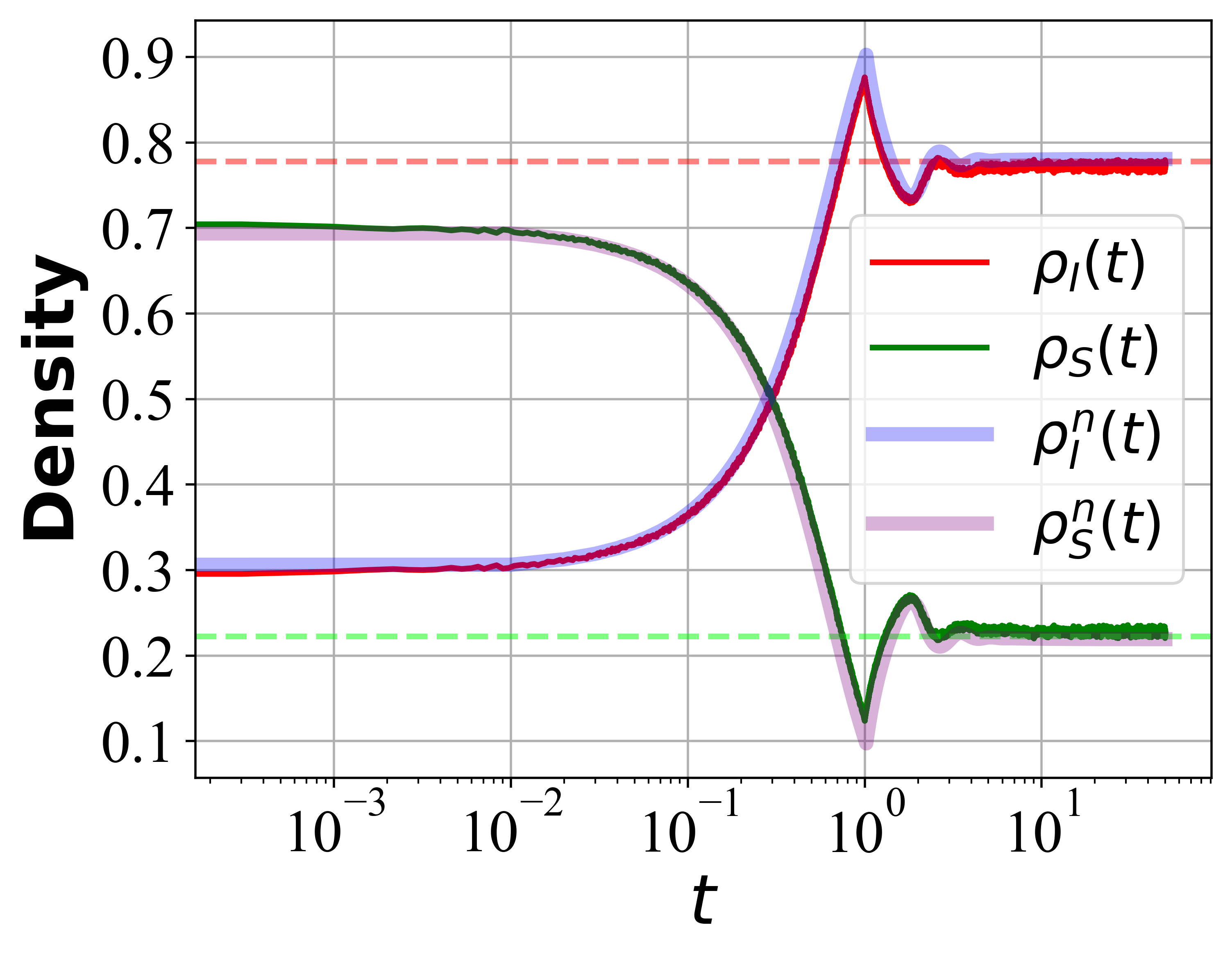}
    }
    \captionof{figure}{
    \textbf{Evolution of infection and susceptibility densities over time for three different recovering waiting time distributions.}  
    The three subfigures show the evolution of infection density $\rho_I(t)$ (red) and susceptibility density $\rho_S(t)$ (green) over simulation time $t$ on a uniform regular graph with degree 10. Subfigure (a) corresponds to an exponential distribution ($\beta = 0.26, \mu = 0.5$), (b) to a lognormal distribution ($\beta = 0.33, \mu = 0, \sigma = 1$), and (c) to a power law distribution ($\beta = 0.3, \lambda = 4, t_0 = 1$).  
    The solid lines represent the average simulation results from 50 independent runs, while the dashed horizontal lines indicate the theoretically predicted steady-state densities for each case. The semi-transparent bands represent the numerical solution of the dynamical equation, with the blue band showing the infection density and the purple band the susceptibility density.  
    The theoretical steady-state infection densities are approximately 0.808, 0.816, and 0.778 for the exponential, lognormal, and power law distributions, respectively.  
    In the figure, it can be seen that the infection density starts at an initial value of 0.3 for all cases and stabilizes near the theoretical steady-state infection density (represented by the horizontal dashed line) around simulation time $t = 10$. Among the three distributions, the exponential distribution reaches the steady-state the fastest and stabilizes, followed by the lognormal distribution. The power law distribution, however, exhibits oscillations around the theoretical steady-state value before stabilizing.  
    The numerical solution of the dynamical equation closely matches the simulation results, though the infection density in the simulation is slightly lower.
} 
    \label{fig:rhoinfty}
\end{figure*}
\begin{figure}
    \centering
    \includegraphics[width=0.8\linewidth]{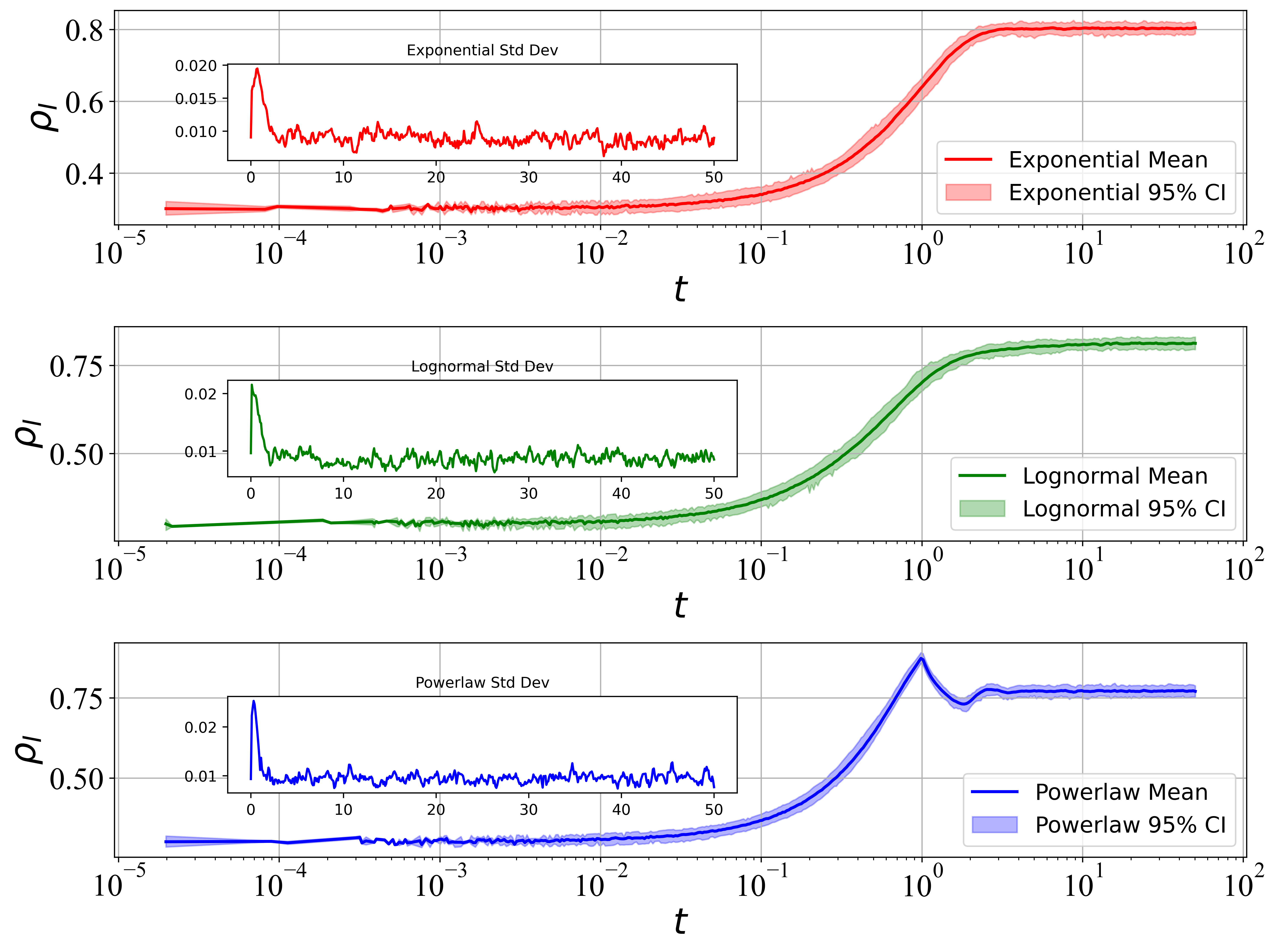}
    \captionof{figure}{ 
    \textbf{Infection Density Evolution and Standard Deviation Comparison for Different Recovery Time Distributions}  
    This figure consists of three subplots, each showing the average infection density over 50 simulations over time, along with the corresponding 95\% confidence interval band (shaded region). The x-axis represents the simulation time, while the y-axis represents the infection density, using a logarithmic scale.  
    From top to bottom, the subplots correspond to simulations using exponential distribution, lognormal distribution, and power law distribution, with the following parameters:
    Exponential distribution ($\beta = 0.26, \mu = 0.5$), Lognormal distribution ($\beta = 0.33, \mu = 0, \sigma = 1$), and Power law distribution ($\beta = 0.3, \lambda = 4, t_0 = 1$).  
    From the figure, it can be observed that the upper bound of the 95\% confidence interval is approximately 0.07. Additionally, each subplot contains an embedded standard deviation plot in the upper right corner, where the x-axis represents simulation time and the y-axis represents the standard deviation of 50 simulation results at each time point. The results show that, except for the initial period, the standard deviation of the simulation results is mostly below 0.01, indicating low variability in the data.  
} 
    \label{fig:std_and_ci}
\end{figure}

\subsection{Final infection density under different $\beta $ values}
\label{sec:4.2}

To explore the effect of varying transmission rates $\beta $ on the steady-state infection densities, we performed multiple independent simulations for each recovery waiting time distribution. For each distribution, we fixed the model parameters and ran simulations over a range of $\beta $ values on uniform regular graphs with degrees $k = 4$, $k = 6$, and $k = 10$, averaging the results from 50 independent runs. The resulting plots, displayed in {Fig.~\ref{fig:rhoinfty-beta}}(a)-(c), illustrate the relationship between $\beta $ and the steady-state infection density $\rho _{I}^{\infty}$ for exponential, lognormal, and power law distributions. Additionally, we aggregated the results for $\langle k \rangle = 10$ from the three distributions and performed a statistical analysis. The corresponding plot, shown in {Fig.~\ref{fig:rhoinfty-beta}}(d), presents the steady-state infection density with 68\% confidence intervals for each distribution.

\begin{figure*}
    \centering
    \subfigure[Exponential]{
        \includegraphics[width=0.45\linewidth]{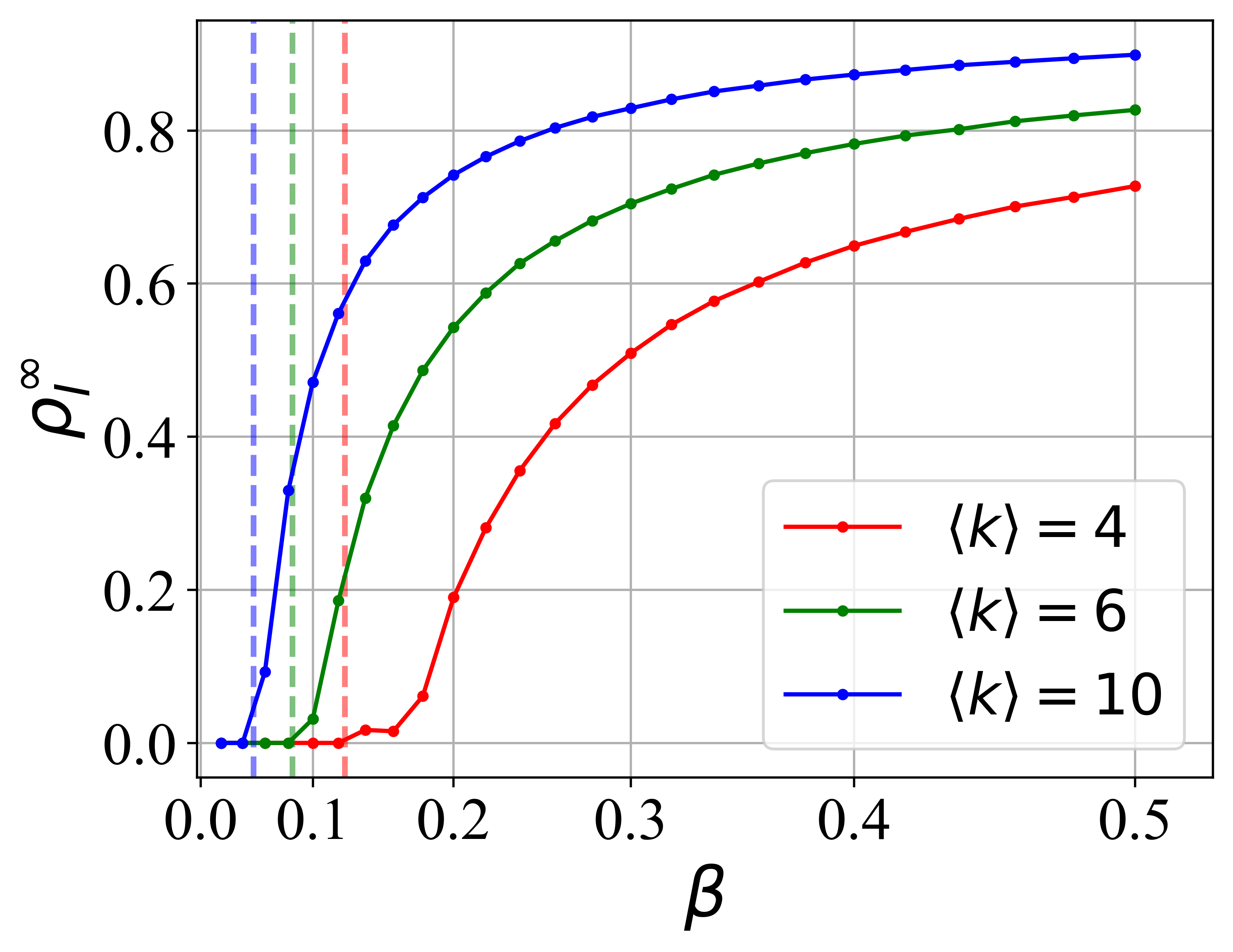}
    }
    \subfigure[Lognormal]{
        \includegraphics[width=0.45\linewidth]{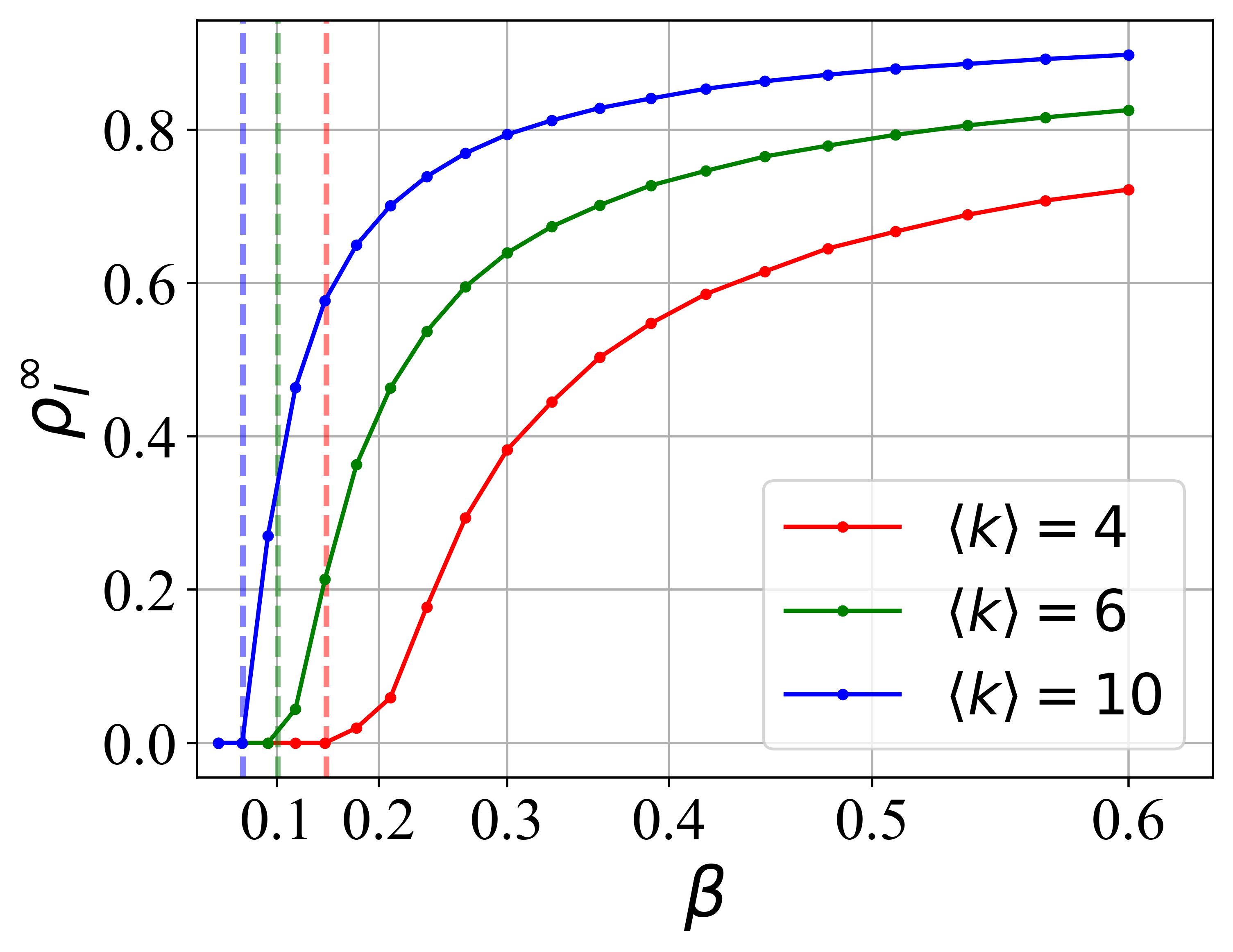}
    }
    \subfigure[Power law]{
        \includegraphics[width=0.45\linewidth]{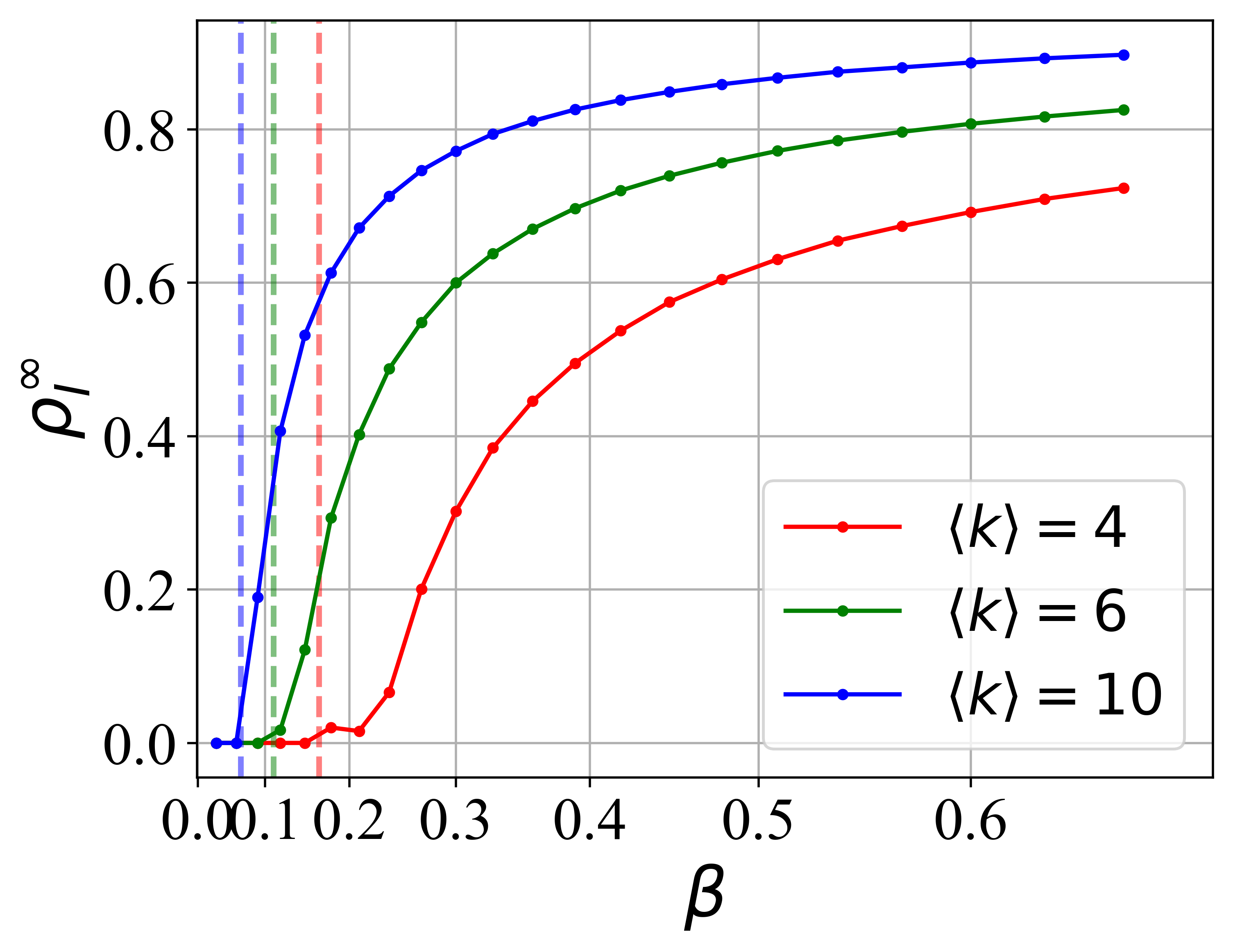}
    }
    \subfigure[Summary]{
        \includegraphics[width=0.45\linewidth]{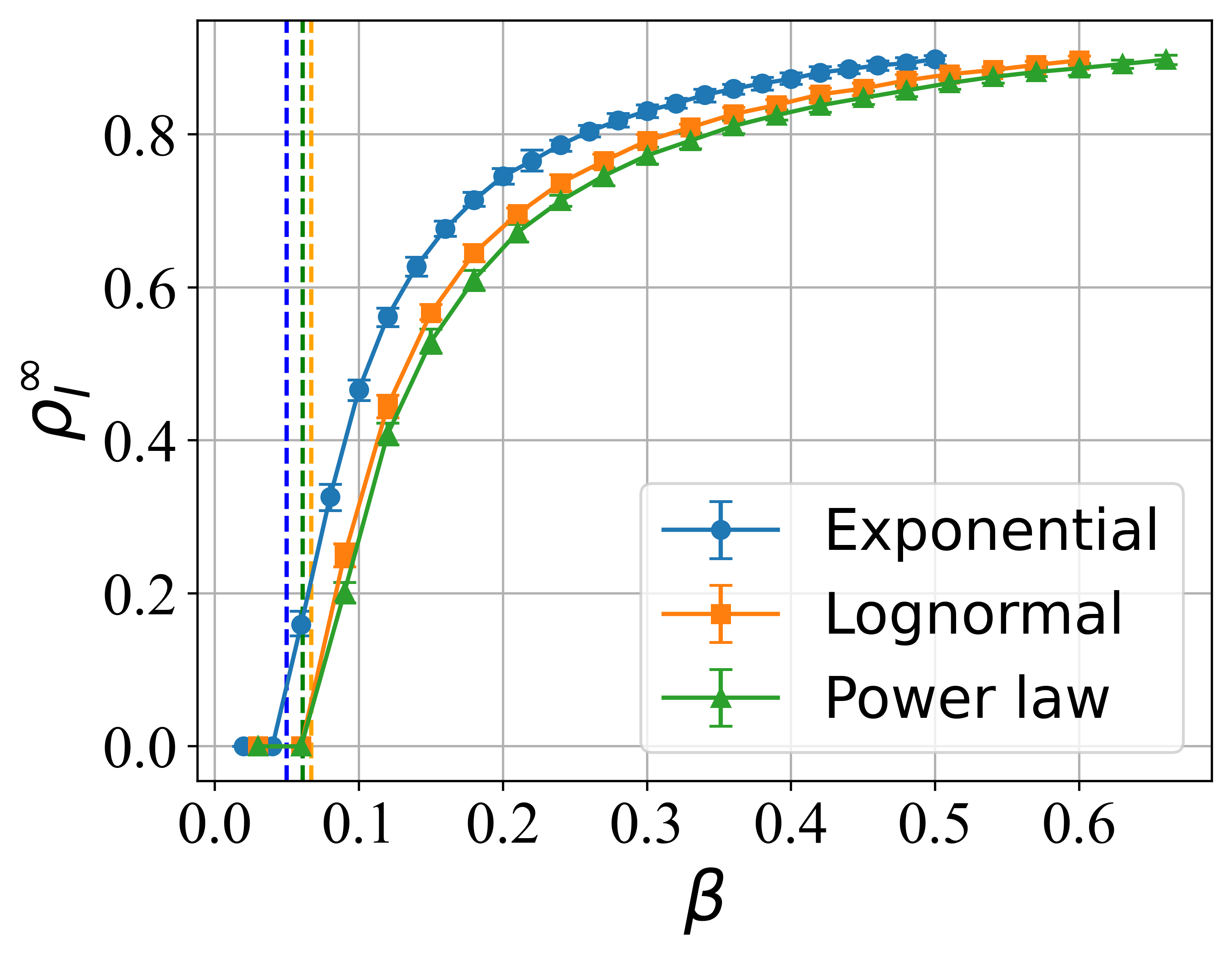}
    }
    \captionof{figure}{
        \textbf{The steady-state infection density $\rho_I^\infty$ as a function of the transmission rate $\beta$ for three different recovering waiting time distributions.}  
        Subfigure (a) shows results for an exponential distribution ($\mu = 0.5$, with 25 values of $\beta$ ranging from 0.02 to 0.5), subfigure (b) for a lognormal distribution ($\mu = 0, \sigma = 1$, with 20 values of $\beta$ ranging from 0.03 to 0.6), subfigure (c) for a power law distribution ($\lambda = 4, t_0 = 1$, with 22 values of $\beta$ ranging from 0.03 to 0.66), and subfigure (d) presents a summary of the three distributions for the case of average degree $\langle k \rangle = 10$. Each subfigure includes results for regular graphs with degrees $k = 4$, $k = 6$, and $k = 10$. The solid lines denote average values over 50 independent simulations, and the vertical dashed lines indicate the theoretical epidemic thresholds, which are approximately 0.05, 0.067, and 0.061 for (a), (b), and (c) respectively. Subfigure (d) combines the results from subfigures (a), (b), and (c) for $\langle k \rangle = 10$ and includes 68\% error bars for each data point, with vertical dashed lines indicating the epidemic thresholds.
        The plots show that the infection density remains close to zero before reaching the theoretical threshold and then starts increasing as $\beta$ increases, with the growth rate first slow, then fast, and finally slow again. For each distribution, the steady-state infection density is higher for graphs with a larger average degree. In subfigure (d), the error bars at each data point are short, never exceeding 0.04 in length. The error bars are longest during the fastest growth phase and shortest when the infection density growth slows down.
    } 
    
    \label{fig:rhoinfty-beta}
\end{figure*}

As shown in {Fig.~\ref{fig:rhoinfty-beta}}(a)-(c), the infection density increases with $\beta $ in all three cases, and networks with higher degrees result in higher infection densities, which is consistent with theoretical expectations. Near the epidemic threshold, we observe a noticeable lag in the onset of sustained infections, with the simulation results showing a delayed threshold compared to the theoretical predictions. This lag becomes more pronounced in lower-degree networks. The SIS model inherently possesses an absorbing state where the number of infections drops to zero. When the gap between the steady state and the absorbing state is small, random fluctuations can push the system into the absorbing state, causing the infection density to initially appear lower in numerical simulations. As the gap widens, the system more closely approaches the theoretical results. Furthermore, in higher-degree networks, the increased number of neighboring nodes reduces the variance in the average number of infections surrounding each node, making the system more resilient to random fluctuations. This is evident in the figure, where the delay effect becomes less pronounced as the degree increases.

Examining {Fig.~\ref{fig:rhoinfty-beta}}(d), we observe that overall fluctuations remain minimal, with the only notable deviations appearing in the regions where the steady-state density grows most rapidly. Even in these regions, the length of the error bars does not exceed 0.04. This result further confirms the accuracy and stability of our simulation methodology, demonstrating its ability to produce precise and reliable steady-state infection densities in nearly all cases.

\subsection{Distribution of infected time $T(\infty )$ at steady state}
\label{sec:4.3}

To simulate the PDF of the infected time $T(\infty )$ at the steady state for infected nodes, we selected several sets of parameters for the three recovery waiting time distributions, fixing other parameters for each run. The infected time at the end of the simulation was taken as an approximation for $T(\infty )$. We then applied the KDE method described in Section~\ref{sec:2.1} to the simulation data, yielding the numerical results for the PDF in {Fig.~\ref{fig:PDF}}.

\begin{figure*}
    \centering
    \subfigure[Exponential]{
        \includegraphics[width=0.45\linewidth]{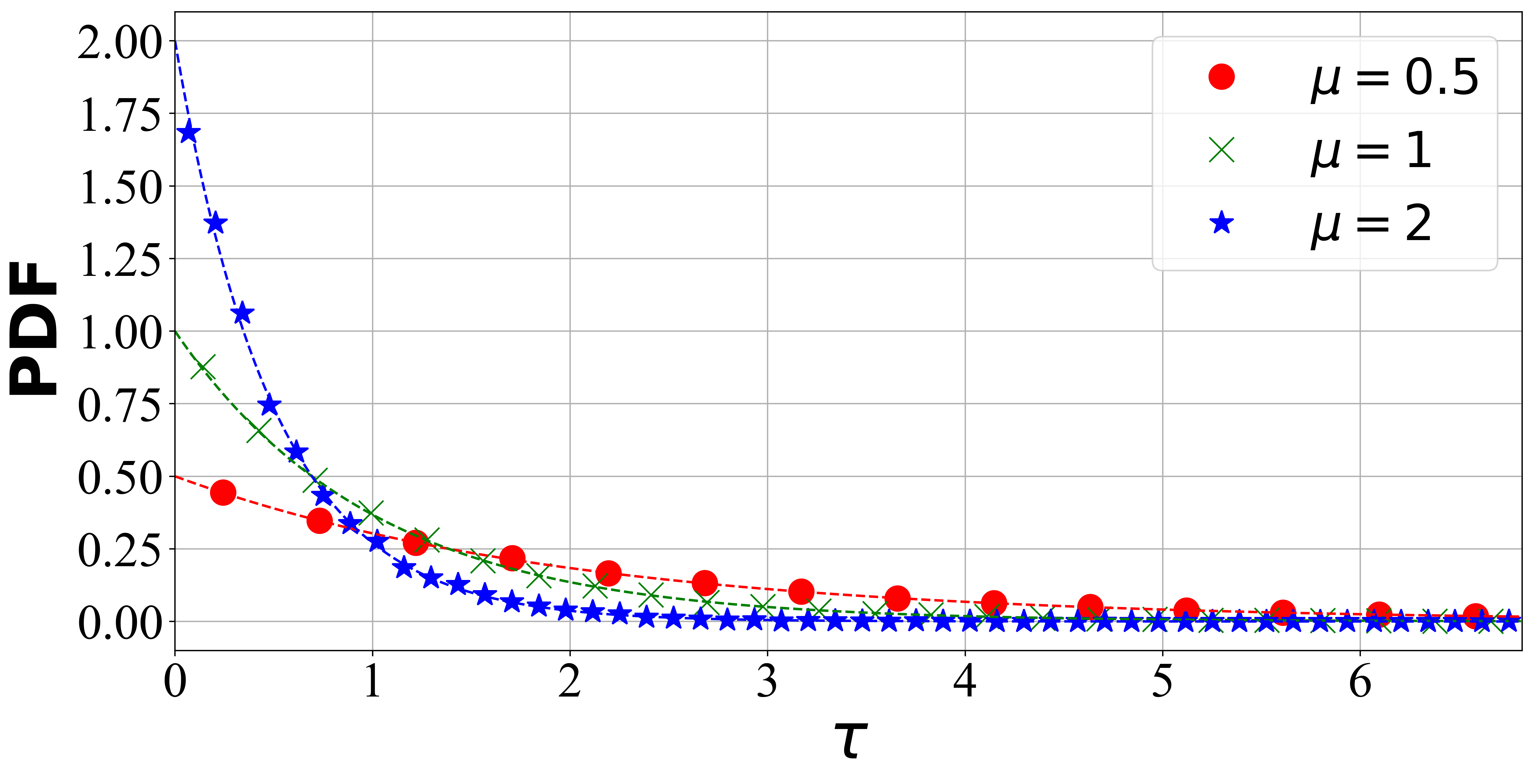}
    }
    \subfigure[Lognormal]{
        \includegraphics[width=0.45\linewidth]{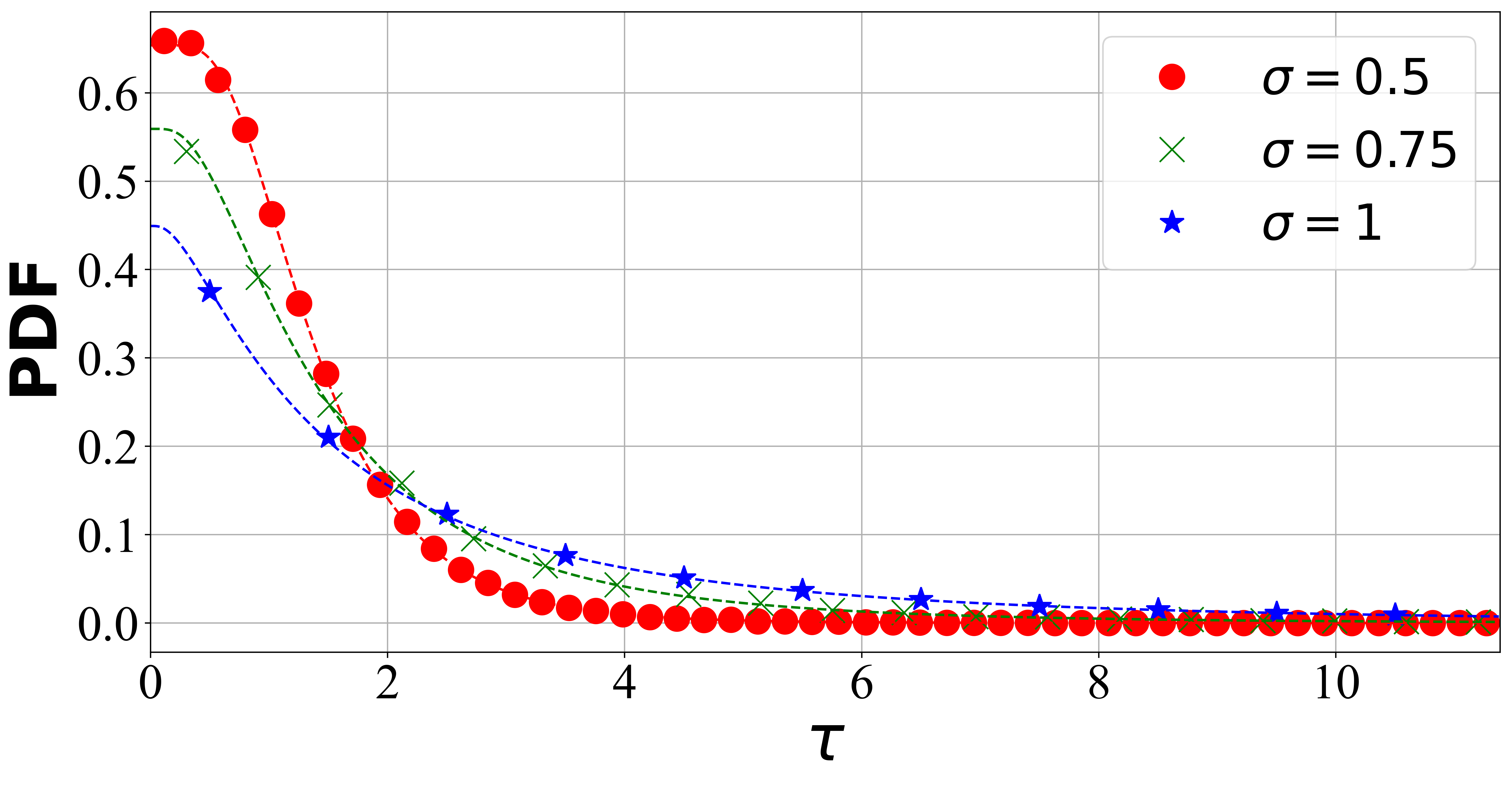}
    }
    \subfigure[Power law]{
        \includegraphics[width=0.45\linewidth]{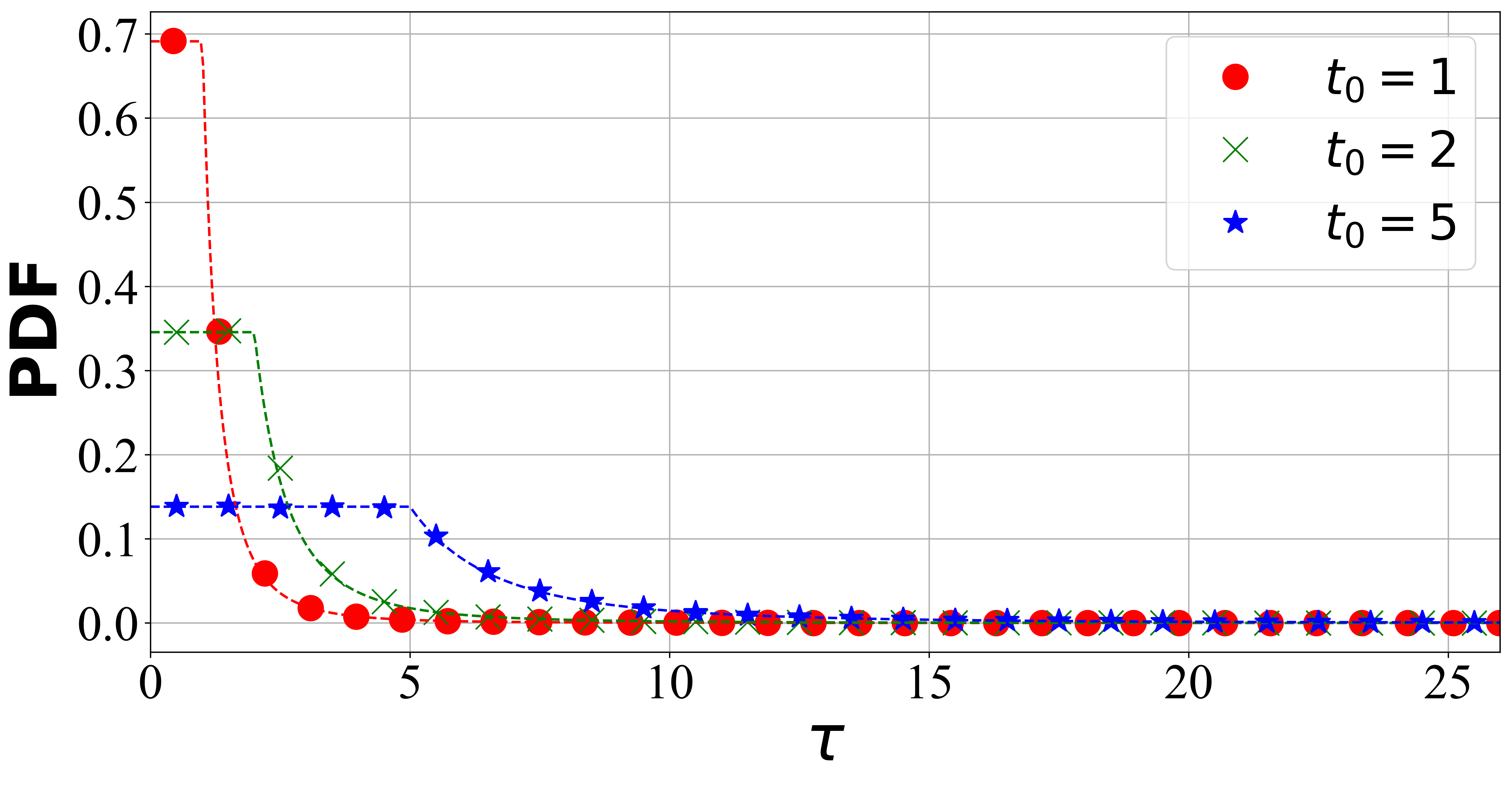}
    }
    \captionof{figure}{
    \textbf{The PDF of the infected time $T(\infty)$ for infected nodes at the steady state of the epidemic.}  
    Subfigure (a) shows the results for an exponential distribution with $\mu = 0.5$, $\mu = 1$, and $\mu = 2$ for $\beta = 0.1$, $\beta = 0.15$, and $\beta = 0.25$, respectively. Subfigure (b) displays results for a lognormal distribution with $\mu = 0.3$ and $\sigma = 0.5$, $\sigma = 0.75$, and $\sigma = 1$ for $\beta = 1$. Subfigure (c) presents the power law distribution with $\lambda = 4.24$ and $t_0 = 1$, $t_0 = 2$, and $t_0 = 5$ for $\beta = 1$.  
    In all subfigures, the solid lines represent the theoretical probability density function (PDF), while the individual data points represent the results of numerical simulations averaged over 50 independent runs on a degree-10 regular graph.  
    The figure shows that for the exponential distribution, the PDF resembles a typical exponential decay, starting at $1/\mu$ and monotonically approaching zero. For the lognormal distribution, the PDF forms a bell-shaped curve with a peak at the mean value and decays symmetrically. In the power law case, the PDF is a horizontal line connecting a monotonically decaying curve. The simulation data points closely match the theoretical PDFs in all cases.
} 
\label{fig:PDF}
\end{figure*}

As seen in {Fig.~\ref{fig:PDF}}, the numerical results closely align with the theoretical curves across all three distributions. To quantify the degree of fit, we computed the Kullback-Leibler (KL) divergence between the numerical and theoretical results. {Table~\ref{table:KL_Divergence}} lists the KL divergence values (rounded to five decimal places) for each parameter set. The relatively low KL divergence values indicate a strong agreement between the numerical simulations and theoretical predictions, providing substantial support for our theoretical framework.

\begin{table*}
  \mytablestyle
  \centering
  \captionof{table}{KL Divergence between Numerical and Theoretical Solutions (rounded to 5 decimal places)}
  \label{table:KL_Divergence}
  \resizebox{\textwidth}{!}{%
  \begin{tabular}{ccc|ccc|ccc}
    \toprule[2pt]
    \rowcolor{cyan!30!white}
    \hline
    \multicolumn{3}{c|}{\textbf{Exponential Distribution}} & \multicolumn{3}{c|}{\textbf{Lognormal Distribution ($\beta=1,\mu=0.3$)}} & \multicolumn{3}{c}{\textbf{Power law Distribution ($\beta=1,\lambda=4.24$)}} \\ \hline
    \rowcolor{lightorange!30!white}
    \multicolumn{1}{c|}{$\beta=0.1,\mu=0.5$} & \multicolumn{1}{c|}{$\beta=0.15,\mu=1$} & $\beta=0.25,\mu=2$ & \multicolumn{1}{c|}{$\sigma=0.5$} & \multicolumn{1}{c|}{$\sigma=0.75$} & $\sigma=1$ & \multicolumn{1}{c|}{$t_0=1$} & \multicolumn{1}{c|}{$t_0=2$} & $t_0=5$ \\ \hline
    \multicolumn{1}{c|}{0.00053} & \multicolumn{1}{c|}{0.00068} & 0.00193 & \multicolumn{1}{c|}{0.00022} & \multicolumn{1}{c|}{0.00091} & 0.00044 & \multicolumn{1}{c|}{0.00470} & \multicolumn{1}{c|}{0.00083} & 0.00177 \\ \hline
    \bottomrule[2pt]
  \end{tabular}%
  }
\end{table*}

Observing the cases in the figure with three different recovery waiting time distributions, we can see that the PDF of the infected time $T(\infty )$ at the steady state for infected individuals is non-increasing across all three types of distributions. This result can be easily explained by {Theorem~\ref{tho:PDF-T}}, which states that $f_{T(\infty )}(\tau )$ is proportional to $\int _{\tau}^{\infty }w(t) \, \mathrm{d}t$. Thus, the non-increasing nature of $f_{T(\infty )}(\tau )$ directly follows from the non-negativity of the PDF.

Additionally, we observe that when the recovery waiting time follows a power law distribution, the distribution of the steady-state infected time $T(\infty )$ is uniform over the interval $\tau \in [0, t_{0}]$, while this uniform region does not appear in other cases. This result is expected: since the power law distribution has a zero probability of taking values within the interval $[0, t_{0}]$, the recovery waiting time of any infected individual will never be less than $t_{0}$. For infected individuals whose infection time has not yet reached $t_{0}$, we cannot distinguish their eventual recovery waiting times based solely on their current infected time. In other words, these individuals are equivalent in terms of their infection time.

In fact, by using {Theorem~\ref{tho:PDF-T}}, we can theoretically prove that, in general, for any interval $E := [m, M]$ with $0 < m < M$, if the PDF of the recovery waiting time $w(t)$ is zero over this interval $E$, meaning that $E$ is an impossible value range for the recovery waiting time, then the distribution of the steady-state infected time $T(\infty )$, conditioned on $T(\infty ) \in E$, will follow a uniform distribution over the interval $E$.

\subsection{Expected infected time $T(\infty )$ at steady state}
\label{sec:4.4}

We analyze the relationship between the mathematical expectation of the infected time $T(\infty )$ at the steady state and the parameters of the recovery waiting time distribution. For this, we conducted numerical simulations using various sets of distribution parameters, keeping other parameters fixed. The infected time at the end of the simulation was used as an approximation of $T(\infty )$, and the mean of these values was calculated as the numerical result for the mathematical expectation, as illustrated in {Fig.~\ref{fig:ET}}.

\begin{figure*}
    \centering
    \subfigure[Lognormal]{
        \includegraphics[width=0.45\linewidth]{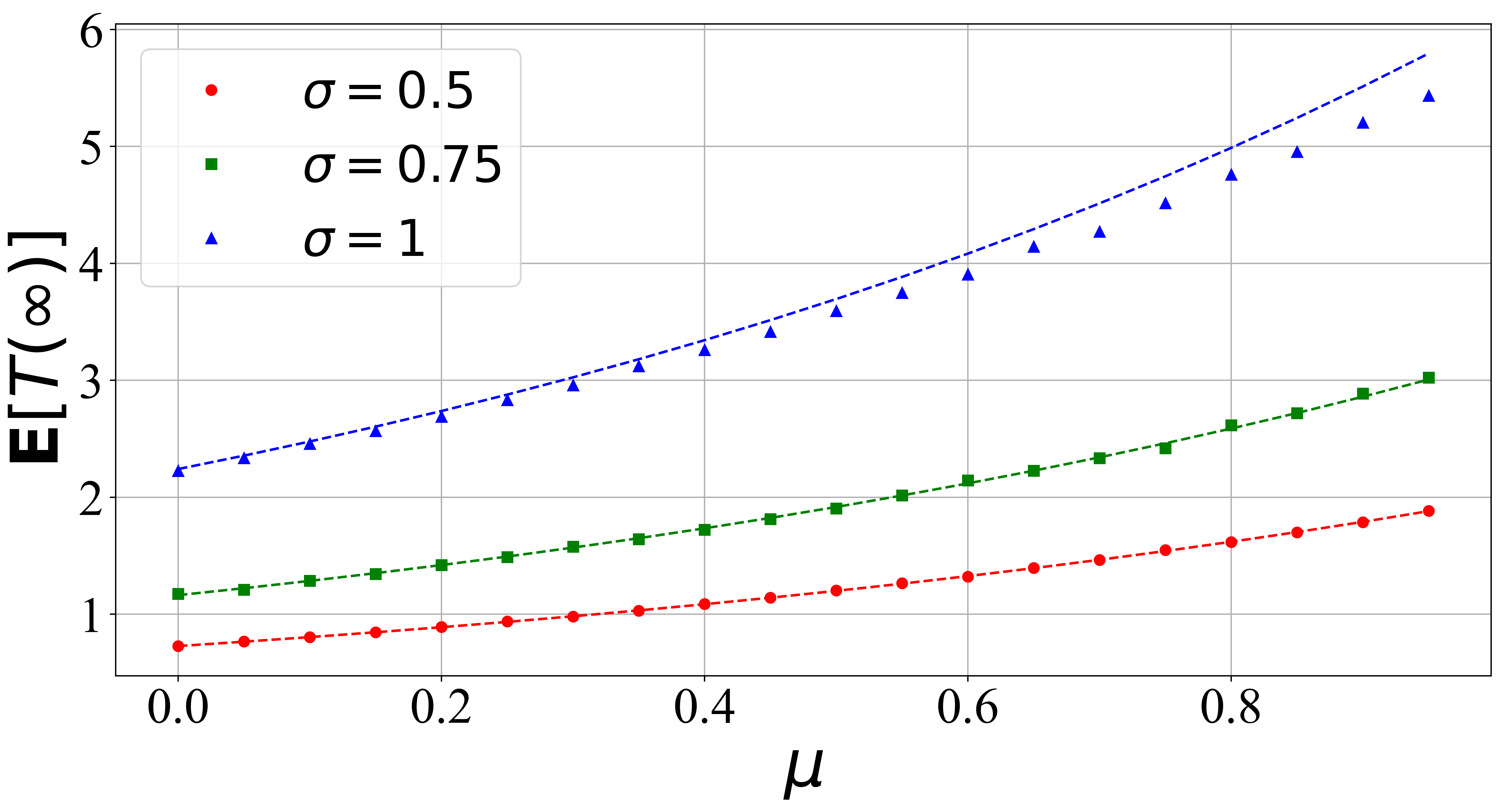}
    }
    \subfigure[Power law]{
        \includegraphics[width=0.45\linewidth]{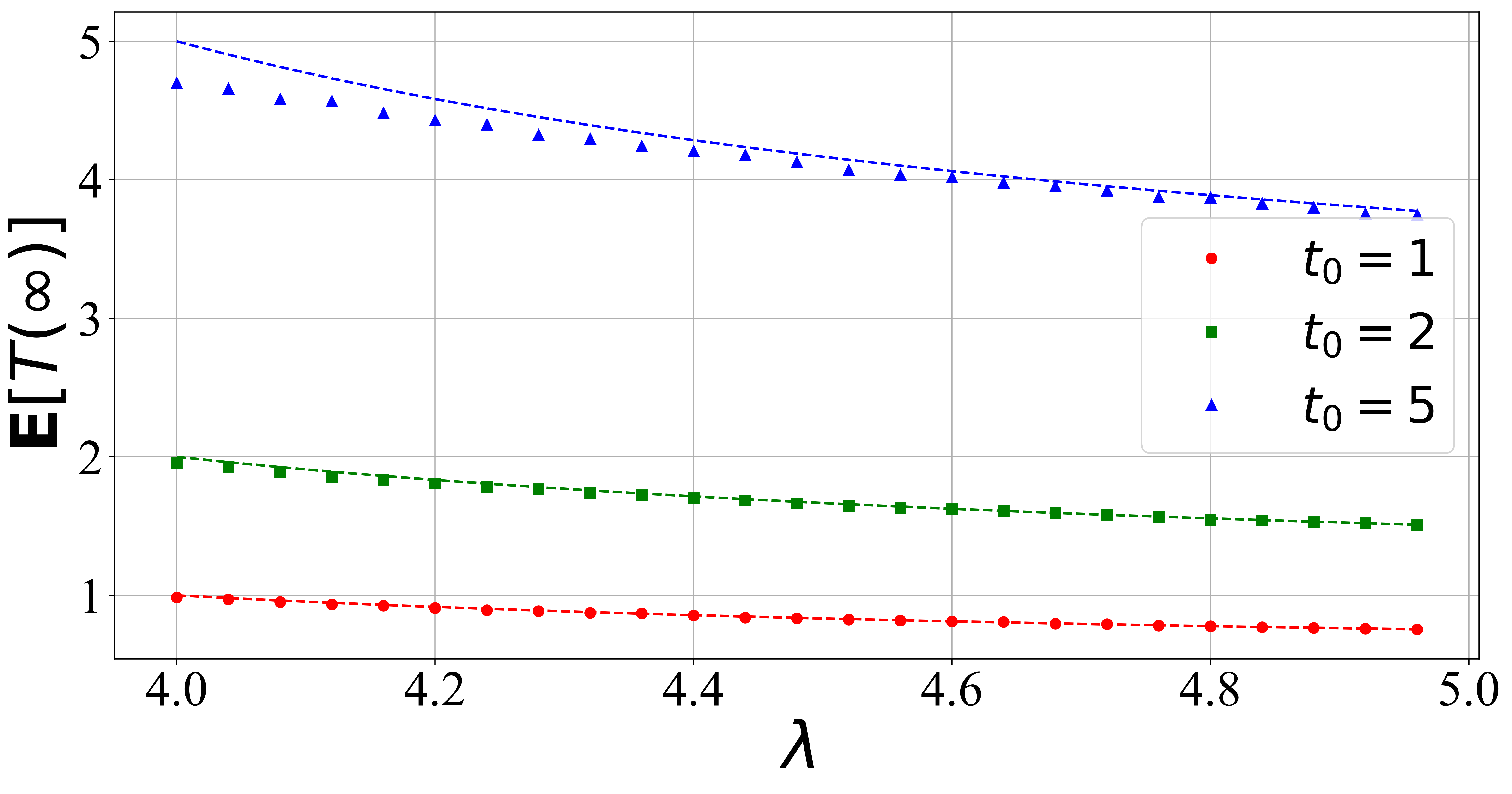}
    }
    \captionof{figure}{
    \textbf{The mathematical expectation of the infected time $T(\infty)$ for infected nodes at the steady state as a function of the recovering waiting time distribution parameters.}  
    Subfigure (a) corresponds to the lognormal distribution with $\sigma = 0.5$, $0.75$, and $1$, while the parameter $\mu$ varies from $0$ to $0.95$ in $20$ equally spaced intervals. Subfigure (b) denotes the power law distribution with $t_0 = 1$, $2$, and $5$, while the parameter $\lambda$ ranges from $4$ to $4.96$ in $25$ equally spaced intervals.  
    In both subfigures, simulations are conducted on a degree-10 regular graph, with $\beta = 1$. Each data point indicates the average result over 50 independent runs. The dashed lines represent the theoretical expectations, and the individual markers indicate the numerical simulation results.  
    In the lognormal distribution case, the expectation increases with $\mu$. For the power law distribution, the expectation decreases with increasing $\lambda$. The simulation results show some minor deviations from the theoretical values. In the lognormal case, discrepancies arise when both $\mu$ and $\sigma$ are large, while in the power law case, deviations occur when $t_0$ is large and $\lambda$ is small. Other data points closely match the theoretical expectations.
} 
    \label{fig:ET}
\end{figure*}

From the figures, we observe that the numerical results align well with the theoretical expectation curves when the distribution parameters are in moderate ranges. However, for more extreme parameter values, there is a slight deviation from the theoretical curves. Upon analyzing the parameter settings for these cases, calculations indicate that the epidemic's steady state approaches the absorbing state. This suggests that the observed deviations stem from the reasons discussed in Section~\ref{sec:4.2}: When the gap between the steady-state of propagation and the absorbing state is narrow, a small random perturbation can cause the system to transition into the absorbing state. In this state, the infected nodes are lost from the system, leading to an expected infected time of zero for those nodes. Consequently, these zero-valued expectations from the absorbing state mix with the expectations from the propagation steady state, resulting in a lowered overall expected infected time compared to the theoretical value. In contrast, when the parameters are such that the gap between the propagation steady state and the absorbing state is wider, the system is less likely to deviate into the absorbing state due to minor perturbations. As a result, the numerical simulations closely match the theoretical predictions under these conditions.

Additionally, as shown in the figure, the expected infected time at the steady state exhibits monotonicity with respect to specific parameters. For the lognormal distribution, the expected infected time in the propagation steady state increases monotonically with $\mu $ and $\sigma $. For the power law distribution, the expected infected time in the propagation steady state decreases monotonically with $\lambda $ and increases monotonically with $t_{0}$. These trends are consistent with the theoretical results.


\section{Conclusions and future work}
\label{sec:5}

In this paper, we introduce the grp-SIS model as an extension of the classic SIS framework, incorporating an arbitrary recovery time distribution. We derive the corresponding mean-field equations, analyze the system's steady-state properties, and validate our results through numerical simulations. Our findings indicate that the recovery time distribution plays a crucial role in shaping disease prevalence, with non-exponential distributions leading to slower epidemic dynamics and a broader infection duration distribution.

This study holds both theoretical significance and practical relevance. First, understanding the distribution of infection times provides critical insights into long-term disease spread trends, enabling public health authorities to develop more effective control strategies. By analyzing infection time characteristics, we can better understand transmission mechanisms, identify high-risk populations, and optimize resource allocation. Furthermore, our findings offer theoretical support beyond epidemiology, particularly in computer science, where similar models can be applied to simulate and predict the spread of network-based viruses.

In terms of real-world applications, our results have significant value in multiple contexts. In public health, understanding the distribution of infection times enables more accurate predictions of epidemic trends, providing a scientific basis for disease prevention and control measures. In medical resource allocation, knowledge of infection time distributions helps optimize resource planning, improving the efficiency and quality of healthcare services. Additionally, during vaccine development and deployment, infection time data can inform vaccine efficacy assessments and guide immunization strategies. Lastly, in cybersecurity, our modelling approach can be leveraged to simulate and predict the spread of malware and other cyber threats, offering theoretical support for network security defense strategies.

Several potential enhancements to the model can be explored in future work. First, while we introduced a general recovery process, the infection process remains Poissonian. Extending the model to incorporate arbitrary infection processes could provide a more comprehensive understanding of disease spreading under heterogeneous conditions. Second, in the numerical simulations presented in Section~\ref{sec:4.2}, deviations from the theoretical results were observed due to the proximity between the steady state and the absorbing state. To mitigate this issue, we suggest employing the quasistationary method \cite{de2005simulate}, which offers a more effective approach to simulating the system while avoiding absorption. Third, our model assumes a homogeneous network, which may limit its applicability to networks with significant degree variations. Addressing this limitation requires analyzing nodes with different degrees separately to extend the grp-SIS model to heterogeneous network systems. By refining both the infection and recovery processes, the model could achieve greater accuracy in predicting disease spread in real-world networks.

\end{document}